\documentclass[sigconf]{acmart}
\usepackage{popets}

\setcopyright{popets}
\copyrightyear{YYYY}

\acmYear{YYYY}
\acmVolume{YYYY}
\acmNumber{X}
\acmDOI{XXXXXXX.XXXXXXX}
\acmISBN{}
\acmConference{Proceedings on Privacy Enhancing Technologies}
\usepackage{amsmath,amsfonts,bm}
\usepackage{amsthm}

\def\eqref#1{equation~\ref{#1}}
\def\1{\bm{1}}

\def\eps{{\epsilon}}

\DeclareMathAlphabet{\mathsfit}{\encodingdefault}{\sfdefault}{m}{sl}
\SetMathAlphabet{\mathsfit}{bold}{\encodingdefault}{\sfdefault}{bx}{n}

\usepackage{multirow}
\usepackage[shortlabels]{enumitem}
\usepackage{booktabs}
\usepackage{nicefrac}
\usepackage{makecell}
\usepackage{subfig}

\usepackage{amssymb}
\usepackage{amsmath,amsthm}
\usepackage{graphicx}
\usepackage{adjustbox}

\theoremstyle{plain}
\newtheorem{theorem}{Theorem}
\newtheorem{lem}[theorem]{Lemma}

\newtheorem{defn}{Definition}

\usepackage[linesnumbered,ruled,vlined]{algorithm2e}
\usepackage{pgfplots}
\usepgfplotslibrary{groupplots}
\pgfplotsset{compat=1.18}
\usepackage{amsmath,amsfonts,bm}
\usepackage{placeins}
\usepackage{array}

\usepackage{tikz}
\usepackage{comment}
\usetikzlibrary{positioning}
\usepackage{float}

\usepackage{xspace}

\providecommand{\name}{{Fre-E2T}\xspace}

\usepackage{pifont}

\providecommand{\todo}[1]{}

\usepackage{colortbl}
\usepackage{tcolorbox}

\newcounter{takeaway}

\definecolor{myblue}{RGB}{31,120,180}
\definecolor{mybblue}{RGB}{166,206,227}
\definecolor{mygreen}{RGB}{251,154,153}
\definecolor{myred}{RGB}{51,160,44}
\definecolor{mybred}{RGB}{178,223,138}

\pgfplotsset{
    colormap={mygreen}{rgb255(0cm)=(254,254,254);rgb255(1cm)=(4,101,53)},
    colormap={myblue}{rgb255(0cm)=(160,30,50); rgb255(1cm)=(39,59,129)},
    colormap={myred}{rgb255(0cm)=(254,254,254); rgb255(1cm)=(160,30,50)},
    colormap={mybred}{rgb255(0cm)=(254,254,254); rgb255(1cm)=(217,10,100)},
    colormap={mybblue}{rgb255(0cm)=(217,10,100); rgb255(1cm)=(10,157,217)},
    colormap={mybgreen}{rgb255(0cm)=(254,254,254); rgb255(1cm)=(150,200,120)},
    colormap={mygreen_r}{rgb255(0cm)=(4,101,53); rgb255(1cm)=(254,254,254)}
}
\definecolor{blue}{RGB}{31,120,180}
\definecolor{bblue}{RGB}{166,206,227}
\definecolor{darkred}{RGB}{139,26,26}
\definecolor{bred}{RGB}{217,10,100}
\definecolor{green}{RGB}{51,160,44}
\definecolor{bgreen}{RGB}{178,223,138}
\usetikzlibrary{patterns}

\usepackage{xparse}
\newcounter{promptcounter}
\NewDocumentCommand{\prompt}{o m}{%
  \refstepcounter{promptcounter}%
  \begin{tcolorbox}[left=5pt, right=5pt, top=3pt, bottom=3pt]
    \textbf{Prompt \thepromptcounter:}%
    \IfNoValueF{#1}{\label{#1}}%
    ~#2
  \end{tcolorbox}%
}

\usepackage{listings}
\lstdefinestyle{queryStyle}{
    basicstyle=\ttfamily\footnotesize,
    breaklines=true,
    breakatwhitespace=true,
    frame=single,
    numberstyle=\tiny,
    stepnumber=1,
    showstringspaces=false,
    tabsize=2,
    captionpos=b,
    backgroundcolor=\color{white},
    rulecolor=\color{black!100},
}

\begin{document}
%==============================================================================

\title{Distributed and Private Textual Data Synthesis from Embeddings}

% Empty location fields are retained as required by the camera-ready template.
\author{Ergute Bao}
\orcid{0000-0002-4438-8065}
\affiliation{%
  \institution{Inria}
  \city{}
  \state{}
  \country{}
 }
\email{baoergute8@gmail.com}

\author{Hongyan Chang}
\affiliation{%
  \institution{MBZUAI}
  \city{}
  \state{}
  \country{}
  }
\email{changhongyan0530@gmail.com}

\author{Ali Shahin Shamsabadi}
\affiliation{%
  \institution{Brave Software}
  \city{}
  \state{}
  \country{}
}
\email{ashahinshamsabadi@brave.com}

\author{Ting Yu}
\affiliation{%
 \institution{MBZUAI}
 \city{}
 \state{}
 \country{}
 }
\email{Ting.Yu@mbzuai.ac.ae}

\author{Xiaokui Xiao}
\affiliation{%
 \institution{National University of Singapore}
 \city{}
 \state{}
 \country{}
 }
\email{xkxiao@nus.edu.sg}

\renewcommand{\shortauthors}{Bao et al.}

\keywords{differential privacy, text synthesis, embedding}

\begin{abstract}
We revisit differentially private (DP) text synthesis in the realistic setting of distributed users, where privacy concerns preclude a trusted curator with access to raw user texts. Existing DP text synthesis pipelines are designed for a trusted, centralized curator and often cannot be deployed in distributed settings due to unrealistic trust and access assumptions; when adapted naively, they require repeated, tightly synchronized user participation and incur significant overhead. To address this gap, we propose a DP--cryptography co-design for textual data synthesis that requires no trusted curator and requires only lightweight user participation. Our approach has two optimized components. First, we design a distributed-friendly DP synthesis algorithm that releases a one-time DP summary in an embedding space: it identifies frequent semantic regions and releases their DP centroids, enabling training-free, non-iterative offline text synthesis. We further introduce semantic support protection, which ensures the released summary avoids semantic neighborhoods of infrequent texts, reducing the risk of exposing rare user data. Second, we develop a custom secure protocol that implements this algorithm over distributed user data, enforcing end-to-end DP guarantees without requiring a trusted curator. On four benchmarks, we achieve utility comparable to the state-of-the-art centralized DP synthesis method.
\end{abstract}

\maketitle

% ----- Main body -----------------------------------------------------------
% The camera-ready main-body limit is 13 pages.
\section{Introduction}\label{sec:intro}

Users generate massive volumes of text on their devices~\cite{apple2025intelli,openai2023gpt4systemcard,casper2023openaiusage,anthropicclaude37sonnetsystemcard2025,braveleohelp2025,geminiteam2023gemini,cheu2025provablyprivateanalyticsinsights,joachims2002optimizing,zhang2020privacysearch}. These texts contain valuable insights into user needs and preferences as well as failure modes of deployed systems. However, collecting and analyzing raw user text poses significant privacy risks~\cite{hitaj2017deep,carlini2021extracting,ding2017collecting,appledpcms2017,Simonite2016,arghire2022samsung,openai2023outage} as texts inherently contain sensitive information.
\looseness=-1

Differentially private (DP) text synthesis~\cite{doi:10.1126/sciadv.abk3283,Ridgeway2021DPChallenge,Freiman2018FormalPrivacySyntheticACS} is a promising approach to obtain population-level insights from user texts. DP~\cite{dmns} offers a robust, worst-case privacy guarantee for individual data against strong adversaries. In the ideal workflow, one releases a \emph{synthetic textual dataset} generated under DP; the synthetic data can then be used for diverse downstream tasks such as data mining or model training without incurring additional privacy risks~\cite{lin2024differentially,Li2021LargeLM,yueetal2023synthetic}.

Most existing DP text synthesis methods are proposed in the centralized setting and typically adopt the following paradigms: (i) \textbf{DP fine-tuning}: Fine-tuning a generative language model (LM) with DP-SGD~\cite{abadi} or its variants on private texts, then sample synthetic texts from the fine-tuned model~\cite{yueetal2023synthetic,flemingsannavaram2024differentially,bommasani2019towards}; (ii) \textbf{API-based private evolution}: Iteratively refining synthetic text candidates generated from public LLM APIs using feedbacks from private texts~\cite{xie2024differentially,chien2026maplemetadataaugmentedprivate}. Such paradigms suffer from four main drawbacks. \looseness=-1

\vspace{0pt}
\noindent\textbf{(1) Trusted curator assumption.}
These approaches are designed for \emph{settings with a trusted curator} that holds the full dataset~\cite{yueetal2023synthetic,chien2026maplemetadataaugmentedprivate,flemingsannavaram2024differentially,bommasani2019towards,xie2024differentially}. In practice, user texts are often distributed and users may not trust a centralized curator with their sensitive data, as evidenced by recent data misuse and technical failures~\cite{arghire2022samsung,openai2023outage}.\looseness=-1

\vspace{0pt}
\noindent\textbf{(2) Repeated user participation in distributed deployments.}
When private texts originate from cross-device populations, straightforward deployments of iterative synthesis pipelines become \emph{multi-round} protocols that require repeated user participation and synchronization~\cite{flemingsetal2024differentially,pretext,hou2025private,xie2024differentially}. Such interactions are often misaligned with cross-device reality, where device availability is intermittent and communication is expensive~\cite{DBLP:journals/corr/KonecnyMR15,pmlrv54mcmahan17a,9084352,bonawitz2019towards}. These constraints are worsened under DP due to the need for sufficient participation per round to offset noise and maintain utility~\cite{pmlrv139kairouz21b,unlocking}.

\vspace{0pt}
\noindent\textbf{(3) Expensive computation and restrictive assumptions.}
DP fine-tuning methods require access to gradients (and often to model parameters)~\cite{yueetal2023synthetic,flemingsannavaram2024differentially,bommasani2019towards,hou2025private,flemingsetal2024differentially}, whereas many capable models are available only through APIs. Even when gradients are available, DP fine-tuning can be prohibitively costly due to gradient computation and communication~\cite{pmlrv139kairouz21b,Li2021LargeLM,liliang2021prefix,hu2022lora,yu2022differentially,jianhaoetal2024promoting,sun2024privatefederateddiscoveryoutofvocabulary}. On the other hand, even though API-based private evolution methods do not compute gradients, they can still be costly---e.g., Aug-PE~\cite{xie2024differentially} generates thousands of candidate
texts per iteration and broadcasts them to every participating user, who must compute nearest-neighbour distances in an embedding space and return a DP vote, blocking the next iteration, which can take hours across tens of iterations. Moreover, existing approaches often rely on \emph{publicly available labels}~\cite{yueetal2023synthetic,xie2024differentially} or \emph{metadata information}  for private texts, which are frequently unavailable or too expensive/sensitive to obtain.

\vspace{0pt}
\noindent\textbf{(4) Limited protection.}
Finally, standard DP ensures on the plausible deniability of membership information~\cite{dpbook,wood2018primer}, but does not explicitly prevent outputs from lying near a particular \emph{infrequent} user data in the embedding space. This can be undesirable, motivating an additional protection tailored to text synthesis.\looseness=-1

\subsection{Our Contributions}
\label{sec:intro-our}

We address the aforementioned gaps through a
holistic approach.We design a DP text synthesis algorithm tailored to the distributed setting, alongside a customized secure computation protocol.

Our algorithm constructs a DP summary via two steps without any costly model fine-tuning: (i) identifying \emph{frequent semantic regions} in the embedding space, and (ii) summarising embeddings within those regions into a privacy-preserving aggregated embedding called \emph{DP centroids}. This DP summary can be post-processed to produce synthetic texts without further privacy loss. 

For step (i), we use the Johnson--Lindenstrauss projection
followed by randomized grid bucketing as a one-shot, training-free
\emph{distributed} clustering, replacing the existing iterative pipelines. For step (ii), we derive a \emph{probabilistic} sensitivity
bound from data-independent randomness, which lets us calibrate DP noise to a high-probability radius
inside each frequent bucket, with the failure event absorbed into the
$\delta$ term of the overall DP guarantee.

Beyond standard DP, we formalize \emph{semantic support protection}. This privacy notion strengthens protection for semantically infrequent user texts by ensuring that, with a high probability, the algorithm output avoids regions near \emph{infrequent} texts in the embedding space (see Section~\ref{sec:semantic-support} for details). We provide theoretical analysis and empirical validation of this property. We emphasize that DP does not directly lead to/violate \emph{semantic support protection}.

We develop a secure distributed implementation of our algorithm in Section~\ref{sec:sol}. Our inherently \emph{distributed-friendly} DP algorithm supports local computation and requires only two phases of user participation. At a high level, each user locally maps their text to an index via sequential embedding, projection, and partitioning, and then shares the encrypted index, with some probability. After users estimate the frequency of their texts, they perform secure, perturbed aggregation of text embeddings associated within the same heavy-hitting group of semantically similar embeddings. The resulting DP centroids enable subsequent text synthesis entirely offline while avoiding model fine-tuning, additional user participation, and further privacy cost. We summarize our contributions as follows.\looseness=-1

\noindent\textbf{(1) New DP algorithm.}
We propose a new training-free DP synthesis algorithm based on a DP summary in an embedding space. 

\noindent\textbf{(2) New privacy notion.}
We introduce \textit{semantic support protection}, which complements DP by limiting proximity to \emph{semantically infrequent} user texts in the embedding space, and provide theoretical and empirical evidence for this property.

\noindent\textbf{(3) Distributed-friendly algorithm design.}
We show that the core computations of the algorithm are simple aggregations admitting an efficient distributed implementation with only \emph{two phases} of user participation. \emph{We do not claim new cryptography protocols.} \looseness=-1

\noindent\textbf{(4) High utility.}
We comprehensively evaluate our approach using four datasets, most of which are not associated with publicly available labels or metadata information. Our algorithm achieves utility comparable to the state-of-the-art \emph{centralized} DP baseline, while offering superior flexibility for deployment in distributed environments and stronger protection for infrequent texts.

\section{Problem Definition}
\label{sec:prob}

\providecommand{\Stag}{\ensuremath{S_{\mathsf{tag}}}\xspace}   
\providecommand{\Ssyn}{\ensuremath{S_{\mathsf{syn}}}\xspace}  

We study differentially private text synthesis. Let $V = \{v_i\}_{i=1}^{N}$ be a dataset of $N$ texts in which $v_i$ is
the single text held by user~$i$ on their own device. Our goal is to release a high-utility
synthetic textual dataset $\widetilde{V}$ while protecting user-level privacy. 

We target settings in which no single party may hold the raw texts:
there is no trusted curator that collects user texts off-device. Instead, we
assume there are two servers
that do not collude. We name the
two servers based on their roles: (i) The \emph{tagging server} \Stag{} assists users in encoding their reports into a random identifier, so that equal inputs receive equal identifiers, without itself learning either the inputs or the identifier. (ii) The \emph{synthesis server} \Ssyn{} receives the encoded user reports, which is then used to produce a DP summary of users, and derives the synthetic texts from it. Section~\ref{sec:agg} explains the two roles concretely. 

\vspace{0pt}
\noindent\textbf{Threat model.} All parties, including both servers and all users, are honest-but-curious: they follow the protocol but may try to infer
information about other users' texts from their views. Neither server is trusted with any user's text or embedding vector. A user
communicates with servers in at most two short phases and never communicate with another user. We discuss potential extensions to malicious settings in Appendix~\ref{app:ext}.

The servers \Stag{} and \Ssyn{} do not collude. It is a common assumption in two-server deployments~\cite{securedpmedian,star,sparsetwoserver,DiscreteGaussian,nebula,hashpruneinvert,SecureML} which can be supported by separate operators or collusion cost due to legal and regulatory constraints~\cite{10.1145/1060590.1060671,10.1007/978354085174528,10.1007/978303122365517}. Every message a user/server sends to a server/user is through an anonymizing proxy, which was also used in prior work~\cite{rappor,prochlo,nebula,sparsetwoserver,DiscreteGaussian,secagg}. 

We assume that both users and servers have black-box access to an embedding model $E$ that maps a text to a $d$-dimensional embedding vector, with $d$ fixed.
In particular, $E$ maps semantically similar texts to nearby vectors.
Access to such a semantics-preserving embedding model was also assumed in prior work~\cite{xie2024differentially}.

We treat the raw private text $v_i$ as a whole piece of user information to protect.
We \textit{do not assume any publicly available label information} associated with $v_i$ that can be used for free, such as domain/category labels for user reviews~\cite{yueetal2023synthetic,xie2024differentially,pretext,hou2025private}. We also do not assume that there is a trained model that extracts metadata information of the user data~\cite{chien2026maplemetadataaugmentedprivate}, since such data may be sensitive and are not used for any model training.

\subsection{Differential Privacy Requirement}\label{sec:privacy-req-DP}

Differential privacy (DP)~\cite{dmns} is commonly regarded as the gold privacy standard for data synthesis~\cite{doi:10.1126/sciadv.abk3283,Ridgeway2021DPChallenge,Freiman2018FormalPrivacySyntheticACS,10.1257/aer.20170627}.

\begin{defn}[Differential Privacy~\cite{dmns}]\label{def:dp}
Let $\mathcal{M}$ be a randomized algorithm with output domain $\text{range}(\mathcal{M})$.
Then $\mathcal{M}$ satisfies DP if, for any neighboring datasets $V$ and $V'$, the probability distributions of $\mathcal{M}(V)$ and $\mathcal{M}(V')$ are similar as quantified by $\varepsilon$ and $\delta$, i.e.,
\begin{align}\label{eq:dp}
\hspace{-0.5em}   \max_{\mathcal{O}\subseteq \text{range}(\mathcal{M})} \Pr[\mathcal{M}(V)\in \mathcal{O}] \le e^{\varepsilon}\cdot\Pr[\mathcal{M}(V')\in \mathcal{O}]+\delta.
\end{align}
Here, $V$ and $V'$ are neighboring datasets if one can be obtained from the other by adding or deleting a single record (i.e., a single text).
\end{defn}

DP provides plausible deniability for the membership of any user text: with or without it, the output distribution is similar.
Smaller $\varepsilon$ and $\delta$ correspond to stronger privacy protection; and vice versa. Every parameter of the DP mechanism we propose, including the randomness it used, is public unless otherwise stated. Such randomness is fixed before any user input and is known to the DP adversary. The probability of an unfavorable draw is absorbed into $\delta$.\looseness=-1

\subsection{Semantic Support Protection (SSP)}\label{sec:semantic-support}
DP does not preclude the mechanism from emitting outputs that are close to a unique or rare user text in the embedding space. The output of a DP mechanism can appear close to an infrequent user text, which  may violate user expectations and can be undesirable~\cite{wood2018primer,acquisti2015privacy}. This motivates us to consider an additional protection tailored to text synthesis, named \textit{semantic support protection}. The aim is to prevent a DP mechanism from emitting outputs near semantically infrequent texts \textit{except with small probability}.
We first define semantically infrequent texts.

\begin{defn}[$(r,t)$-Semantically Infrequent]\label{def:euclidean-secret}
Given $r>0$ and positive integer $t$, textual dataset $V$, and embedding model $E$, we say a data point $v^*\in V$ is $(r,t)$-semantically infrequent (for $V$ and $E$) if fewer than $t$ data points from $V$ are within distance $r$ of $v^*$ in the embedding space:
\begin{align}
    |\{j\in[N]\textrm{ such that } \|E(v_j) - E(v^*)\|_2 \le r \}| < t.\label{eq:frequent}
\end{align}
\end{defn}

We denote the set of $(r,t)$-semantically infrequent and frequent texts in $V$ as $V_{\text{infre}}(r,t)$ and $V\setminus V_{\text{infre}}(r,t)$, respectively.

\begin{defn}[$(r,t,r^*,\delta^*)$-Semantic Support Protection]\label{def:semantic-protect}
Let $\mathcal{M}$ be an $(\varepsilon,\delta)$-DP mechanism that outputs a set of texts $\widetilde V$.
We say that $\mathcal{M}$ satisfies $(r,t,r^*,\delta^*)$-semantic support protection (SSP) if for any $(r,t)$-semantically infrequent $v^*\in V_{\text{infre}}(r,t)$,
\begin{align}
    \Pr\Big[\exists\, \widetilde v \in \widetilde V \textrm{ such that } \| E(\widetilde v) - E(v^*)\|_2 \le r^*\Big] \le \delta^*,
\end{align}
where $r^*>0$, and $\delta^*\in(0,1)$ is the failure probability.
% The probability is over the randomness of $\mathcal{M}$.
\end{defn}

\noindent\textbf{Difference from DP, LDP, and metric-DP.} This definition provides an additional privacy layer beyond DP~\cite{dmns} and local DP~\cite{privately} by explicitly limiting the probability that released outputs lie within radius $r^*$ of any semantically infrequent input in the embedding space. As we will see in Figure~\ref{fig:semantic}, a standard DP method for text synthesis such as Aug-PE~\cite{xie2024differentially} does not automatically achieve SSP. Indeed, DP or LDP permits events such as ``the output lies close to an infrequent input text'' to occur with non-negligible probability as long as the probability ratio under neighboring datasets is bounded.  SSP explicitly complements this limitation. Metric-LDP~\cite{Geoind,cgdp,tokenperturb,wordleveldp} guarantees that outputs cannot distinguish embeddings that are close in metric, such as the $\mathcal{L}_2$ distance measured on the embedding space. Similar to DP and LDP, metric-LDP does \emph{not} explicitly bound the probability of emitting a
value near a rare user input. \looseness=-1

At first sight, requiring the released output to \emph{avoid} a neighbourhood
of an infrequent input $v^*$ may look like it could distinguish the
dataset $V$ that contains $v^*$ from the neighbour $V'=V\setminus\{v^*\}$
that does not. The key observation is that: if $v^*$ is
$(r,t)$-infrequent in $V$ then it is at least $(r,t-1)$-infrequent in $V'$,
and conversely a point that is infrequent under $V'$ remains infrequent
under $V$. Definition~\ref{def:semantic-protect} therefore applies to the
\emph{same} $v^*$ under both $V$ and $V'$; and the neighborhood location $E(v^*)$ is never released. 

\noindent\textbf{Simplified definition.} Setting $r^*=r$, we obtain the following simplified SSP.

\begin{defn}[$(r,t,\delta^*)$-Semantic Support Protection]\label{def:semantic-protect-simplified}
Mechanism $\mathcal{M}$ satisfies $(r,t,\delta^*)$-SSP if for any $v^*\in V_{\text{infre}}(r,t)$,
\begin{align}
    \Pr\Big[\exists\, \widetilde v \in \widetilde V \textrm{ such that } \|E(\widetilde v)  - E(v^*)\|_2 \le r\Big] \le \delta^*.
\end{align}
\end{defn}

\section{Related Work}\label{sec:rel}
As we have mentioned in Section~\ref{sec:intro}, DP text synthesis has been studied primarily in the \emph{centralized} setting with a trusted curator holding the full private dataset~\cite{yueetal2023synthetic,chien2026maplemetadataaugmentedprivate,flemingsannavaram2024differentially,bommasani2019towards,xie2024differentially,liu2025uraniadifferentiallyprivateinsights}. Adapting the existing paradigms to distributed settings has been explored in recent work, but existing solutions typically retain at least one of the following requirements: gradient/parameter access to an LM for training, publicly available labels/categories for private texts, or \emph{multi-round} synchronized user and server participation.
For example, methods that update or optimize models using user feedback, such as POPri~\cite{hou2025private} and federated DP fine-tuning variants~\cite{flemingsetal2024differentially}, require gradient access and repeated rounds with tight synchronization between users and the server. Similarly, distributed adaptations of private evolution~\cite{pretext} require users to repeatedly process server-provided public candidates and return DP-perturbed signals in each round.
Such designs can be difficult to scale in cross-device environments, where user availability is intermittent and communication is costly~\cite{DBLP:journals/corr/KonecnyMR15,pmlrv54mcmahan17a,9084352,bonawitz2019towards}, and these constraints are further amplified under DP due to the need for sufficient participation per round to offset noise~\cite{pmlrv139kairouz21b,unlocking}.

A complementary line of work directly perturbs token, sentence, or document
embeddings under \emph{metric-LDP} (also called
$d$-privacy), so that the indistinguishability guarantee scales with a
distance in embedding space rather than being uniform across all
pairs~\cite{wordleveldp,tokenperturb,klymenkoetal2022differential,matternetal2022differentially}. We target the standard
\emph{text-level} $(\varepsilon,\delta)$-DP rather than a distance-scaled
guarantee. 

Beyond formal DP mechanisms, recent works have explored privacy-preserving analysis of user text corpora via layered, non-DP protections. For example, Clio~\cite{tamkin2412clio} uses LLM-based summarization and embedding-based clustering to surface aggregated usage patterns, and then applies multiple \emph{privacy layers} such as PII-avoiding prompts and automated auditing to reduce exposure of private information. Clio focuses on the output utility rather than worst-case privacy guarantees such as DP, making it prone to attacks~\cite{annamalai2026cliopatraextractingprivateinformation}. There are also some other work in the literature, e.g., secret protection~\cite{wang2025secretprotectedevolutiondifferentiallyprivate,ganesh2025hushprotectingsecretsmodel}, privacy with respect to sensitive references~\cite{vinod2025invisibleinkhighutilitylowcosttext}, and control codes~\cite{zhao2025controlledgenerationprivatesynthetic}. These methods often rely on assumptions about sensitive information and prior knowledge regarding the data distribution, and sometimes the computation power of an adversary.

% \vspace{0pt}
% \noindent\textbf{Our positioning.}
% Compared with prior DP fine-tuning and private evolution approaches, our main contribution is a new \emph{non-iterative}, training-free DP text synthesis mechanism based on releasing a one-time DP summary in an embedding space, together with the new privacy notion of \emph{semantic support protection}.
% The mechanism is inherently \emph{distributed-friendly}: its core computations are simple aggregations that admit an efficient distributed implementation with only two phases of user participation.

\begin{figure*}[t!]
\centering
\includegraphics[width=0.95\textwidth]{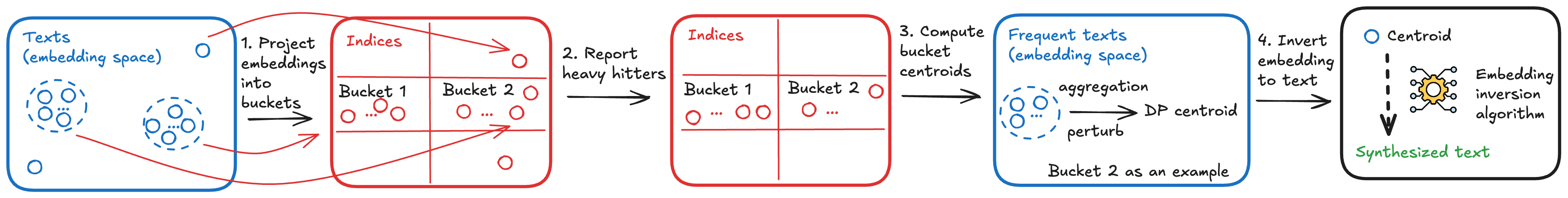}
\caption{Overall idea of \name, illustrated in the centralized setting.}
\Description{Text embeddings are projected into grid buckets. Frequent buckets are selected, their embedding centroids are aggregated and perturbed, and an embedding inversion model converts the private centroids into synthetic texts.}
\label{fig:centralized-illustration}
\end{figure*}

\section{Synthesizing Texts via Frequent Embeddings}\label{sec:idea}
We present our DP algorithm for synthesizing texts, named \name. \name achieves formal differential privacy guarantees for individual texts (Definition~\ref{def:dp}) and also provide semantic support protection for infrequent texts (Definition~\ref{def:semantic-protect}). Section~\ref{sec:sol} gives the two-server
protocol that realizes the functionality of Fre-E2T under the threat model introduced in Section~\ref{sec:prob}.

To reduce the number of traversals on input and therefore enable a lightweight distributed implementation, we first sketch the user texts as a summary, and then synthesize texts from the summary.

\vspace{0pt}
\noindent\textbf{Clustering embedding vectors.} A canonical method for summarizing data is clustering~\citep{Lloydcluster,MacQueen1967KMeans,Xu2005SurveyClustering}. In clustering, similar data points are grouped together based on some measure. Such a group of data points is referred to as a cluster. The goal of clustering is to find representative structures to summarize the input dataset. A common approach for accomplishing that is to ensure that each input data point lies within a ball of some radius centered at a centroid, the Euclidean average of all points within the same cluster. We could apply this approach to text embeddings and obtain cluster centroids for input texts. 

\vspace{0pt}
\noindent\textbf{Inverting embeddings to texts.} For each centroid (a vector in $\mathbb{R}^d$), we can then use a function/model that inverts embedding vectors into textual data~\cite{morris-etal-2023-text,luo2025promptinferenceattackdistributed} to generate high-utility synthetic texts. Intuitively, if each centroid is near some data points, then the texts synthesized from the centroid should also bear similarity to the original texts corresponding to those data points.\looseness=-1 

\vspace{0pt}
\noindent\textbf{Technical Challenges.} Existing DP clustering methods generally rely on an iterative approach to refine centroids, based on their distances from each data point $x_i$: each iteration traverses the whole dataset in order to evaluate the distances from each data point to all the centroids, followed by updating the centroids to maximize the overall utility~\citep{dpem,sulq,DPLloyd,maxcoverdp}. To build a more distributed-friendly algorithm, we avoid such \emph{iterative updates} and computing \emph{pairwise distances} between data points, which would require user-user communication in the distributed setting. Furthermore, when computing the centroid for each cluster, standard worst-case DP analysis~\cite{abadi} requires an additive random noise whose scale is proportional to the \emph{maximum norm} of any data point, which is not ideal. Intuitively, we should be able to utilize the structural information of clusters to reduce the amount of noise needed for DP.

\subsection{The Complete Algorithm}\label{sec:idea-algo}
Our algorithm named \name (\textsf{Frequent Embedding to Text}) consists of four sequential steps, presented as Algorithm~\ref{alg:centralized} and Figure~\ref{fig:centralized-illustration}. \looseness=-1

\vspace{0pt}
\noindent\textbf{Step 1: Projection and Partitioning.} 
We first obtain the embedding vector \(x_i\in \mathbb{R}^d\) for each text. Each embedding vector $x_i = E(v_i)\in\mathbb{R}^d$ is obtained via an embedding function $E$. The entire matrix of embedding vectors is written as $X \in \mathbb{R}^{N \times d}$. Next, we project each vector to $\mathbb{R}^k$, by computing $y_i=\frac{1}{\sqrt{k}} x_iZ\in \mathbb{R}^k$. Here, $Z$ is a random $\mathbb{R}^{d\times k}$ matrix where each entry is sampled independently from the standard Gaussian distribution. Pairwise \(\mathcal{L}_2\) distances are preserved up to a small multiplicative distortion with high probability for Gaussian projectors, when $k$ is large enough, by the Johnson Lindenstrauss Lemma~\cite{DasguptaGupta2003,JohnsonLindenstrauss1984}. In other words, two records are close in \(\mathbb{R}^k\) if and only if their originals are close in \(\mathbb{R}^d\). As a result, data points with a high support in their neighborhood in $\mathbb{R}^d$ will still have a high support in $\mathbb{R}^k$. The Euclidean average of data points in the same cluster can thus serve as a good representation of them and frequent texts will be placed in a cluster with many close neighbors. 

To group nearby projected points without computing all pairwise distances, we partition the $k$-dimensional space into ``buckets''. To do that, we first discretize \(\mathbb{R}^k\) with a randomly shifted grid of edge length \(L>0\). In particular, for each coordinate \(j\in\{1,\dots,k\}\), a public random offset \(
o_j\) is sampled from \(\mathrm{Unif}[0,L).\)
The grid boundaries are placed at \(o_j + mL\) for all \(m\in\mathbb{Z}\).  The \(k\)-dimensional bucket index for user \(i\) is the integer vector
\(\textstyle
b_i \;=\; \bigl(b_{i,1},\dots,b_{i,k}\bigr),\) with \(\textstyle
b_{i,j} \;=\; \left\lfloor \frac{y_{i,j}-o_j}{L} \right\rfloor \in \mathbb{Z}.\) We refer to Figure~\ref{fig:grid-illustration} for an illustration. $b_i$ is then randomly permuted to obtain $h_i$. Vectors that are close to each other in $\mathbb{R}^k$ are likely to end up with the same index before and after the permutation; and vice versa~\cite{lsh1,lsh2,lsh3}. 

The edge length $L$ is fixed a priori as a public design parameter of the protocol. Choosing $L$ controls the granularity of the partitioning and clustering: a smaller 
$L$ yields finer buckets, making it less likely that two nearby points fall into the same bucket (thus producing more, smaller clusters), while a larger $L$ coarsens the grid and increases the probability that points which are potentially farther apart in $\mathbb{R}^d$ collide in the same $k$-dimensional bucket. As we will see later, $L$ also influences the privacy guarantees. Consequently, the appropriate choice of $L$ is application-dependent and should be guided by the desired trade-off between resolution and the privacy protection for user texts. 

\vspace{0pt}
\noindent\textbf{Step 2: DP heavy hitters.}  With the buckets assigned, we then count the number of data points in each bucket. For each data point, we flip a coin to decide whether it contributes to the corresponding bucket and filter out buckets with counts $<\tau$. This sample-and-threshold process separates heavy-hitting bucket indices from those that have not accumulated enough counts while introducing enough randomness to ensure DP~\cite{sandtdp} and also has an efficient distributed implementation~\cite{nebula}. Our design has avoided injecting random noises to each user's embedding vector directly as in local DP~\cite{privatemining,privately}, which would introduce a large perturbation to the results, compromising utility.

\vspace{0pt}
\noindent\textbf{Step 3: DP Centroid computation.} Now that we have identified the ``heavy-hitting buckets'', we can then compute the average of the original data points within each of these buckets in $\mathbb{R}^d$. In particular, each text that has contributed to a heavy-hitting bucket index (with $\text{coin}_i=1$) contributes to the vector sum; those that are not sampled during the last step are not counted. Additive random Gaussian noise is injected to the sum. For DP analysis, we avoid considering the worst-case sensitivity of the vector sum, which is the radius of the whole $d$-dimensional space. Instead, we rely on the crucial property that only data points that are close to each other can reside in the same cluster, except with a small error probability. Therefore, we only need to inject DP noise that is proportional to the radius of the cluster rather than the whole space. 

\vspace{0pt}
\noindent\textbf{Step 4: Embedding to text conversions.} Finally, from the centroids of each heavy-hitting bucket, we invert it to obtain synthesized texts. DP is preserved under post-processing~\cite{dmns}. 

\begin{algorithm}[t]
\DontPrintSemicolon
\KwIn{Text entries $\{v_1, \ldots, v_N\}$; embedding function $E$; random projection matrix $\mathbf{Z}$; bucket edge length $L$; random offsets $\{o_j\}_{j=1}^k$; random hash function $H$; subsampling rate $p_s$; threshold $\tau$; noise scale $\sigma$.}
% \KwOut{Set of synthetic texts $\mathcal{O}$.}
Initialize $\mathcal{B} \leftarrow \emptyset$, $\mathcal{S}\leftarrow \emptyset$, $\mathcal{O}\leftarrow \emptyset$ \;
\tcp{\textbf{1. Projection and Partitioning}}
\For{$i = 1, \ldots, N$}{
    Compute embedding $x_i \gets E(v_i) \in \mathbb{R}^d$ \;
    Compute projection $y_i \gets \frac{1}{\sqrt{k}} x_i\mathbf{Z} \in \mathbb{R}^k$ \;
    Identify bucket index $b_i$ based on $y_i$, $\{o_j\}_{j=1}^k$, $L$ \;
    Randomly permute $h_i=H(b_i)$.
}

\tcp{\textbf{2. DP Heavy Hitters}}
\For{$i = 1, \ldots, N$}{ 
    \If{$\text{coin}_i\sim \text{Bernoulli}(p_s)$ is $1$}{
        Add $x_i$ to bucket $\mathcal{B}[h_i]$\;
    }
}
\For{each bucket index $h$ in $\mathcal{B}$}{
    Count $n_h = |\mathcal{B}[h]|$ \;
    \If{$n_h < \tau$}{
        Remove bucket index $h$ from $\mathcal{B}$\;
    }
}

\tcp{\textbf{3. DP Centroid Computation}}
\For{each bucket index $h$ in $\mathcal{B}$}{
    Compute sum $S_h \gets \sum_{x_i \in \mathcal{B}[h]} \big(\text{coin}_i\cdot x_i\big)$ \;
    Compute  $\widehat{\mu}_h \gets \frac{1}{n_h} (S_h + \mathcal{N}(0, \sigma^2 \mathbf{I}_d))$\;
    Add $\widehat{\mu}_h$ to $\mathcal{S}$ \;
}
\tcp{\textbf{4. Embedding to Text Conversion}}

\For{each centroid $\widehat{\mu}_h$ in $\mathcal{S}$}{
    Synthesize text $\widetilde{v}_h \gets \text{Invert}(\widehat{\mu}_h)$\;
    \textbf{Output} $\widetilde{v}_h$
}
% \KwRet $\mathcal{O}$ \;
\caption{Fre-E2T
% : Algorithm for Centralized Differentially Private Text Synthesis
}
\label{alg:centralized}
\end{algorithm}

\section{Privacy Analysis}\label{sec:privacy}
\subsection{Differential Privacy Guarantees}\label{sec:privacy-dp}
Privacy loss comes from two components, identifying heavy hitting buckets and computing the centroid within each heavy hitting bucket. Without loss of generality, we work on the matrix of embedding vectors, written as $X$. The $i$-th row of $X$ is associated with text $i$. We consider a neighboring input $X'$ such that $X$ has an additional embedding vector $x$. We condition on the event that inputs $X$ and $X'$ work with the same projection matrix $Z$ and the same hash function $H$. We let the permuted bucket index of $x$ be $h$.

\noindent\textbf{Heavy Hitters.}
For the truncated heavy hitter histogram, it suffices to analyze the count of index $h$, when it reaches threshold $\tau$. 

\begin{lem}[Sample-and-threshold privacy~\cite{sandtdp}]\label{lem:s-and-t-dp} With $q:=1-(1-p_s)\,\exp(-\varepsilon_{fre})$ and \(\textsc{D}\big(q\,\|\,p\big):= q\log\frac{q}{p} + (1-q)\log\frac{1-q}{1-p}\), the heavy hitter histogram $\{(h,n_h)\}$ in Alg.~\ref{alg:centralized} satisfies $(\varepsilon_{fre},\delta_{fre})$-DP with
\begin{align}
\varepsilon_{fre}\leq \log\frac{1}{1-p_s}\,, \textrm{ and } \quad\delta_{fre}\leq \exp\left(-\frac{\tau}{q}\cdot\textsc{D}\big(q\,\|\,p_s\big)\right). 
\label{eq:dp-fre-delta}   
\end{align}
\end{lem}

The sum of all random contributions of $X$, denoted as $n_h$, therefore follows the Bernoulli distribution $Bin(N_h,p_s)$, where $N_h$ is the overall number of users with index $h$. For its neighboring dataset, this distribution becomes $Bin(N_h-1,p_s)$. Lemma~\ref{lem:s-and-t-dp} bounds this divergence, with a high probability.

\noindent\textbf{Choosing $\tau$.} The threshold $\tau$ is a public
parameter that controls which buckets are declared heavy hitters.
We set $\tau=p_s\cdot t$, so that in expectation a bucket with strictly
fewer than $t$ neighbors within distance $r$ contributes at most $p_s t$
sampled submissions and is therefore filtered out. Larger $\tau$ leads to fewer heavy hitting buckets and smaller $\delta_{fre}$ via Lemma~\ref{lem:s-and-t-dp}.\looseness=-1

\noindent\textbf{Heavy hitting Indices.} The indices of heavy hitters correspond to frequent embedding vectors. These indices are conceptually close to labels of user texts (e.g., ``Texas barbecue joints''), which were utilitized in prior works~\cite{yueetal2023synthetic,xie2024differentially}. In our setting, such information is not publicly available before hand. Here, the index of a heavy-hitting bucket index is conceptually similar and obtained in a privacy-preserving manner, albeit being less informative (since we cannot interpret the indices themselves). 

\noindent\textbf{Computing Centroids.} Again, we consider any bucket $h$ with count $n_h\geq\tau$. There are two cases. If the additional data point in $X$, denoted as $x$, is not added to the bucket (due to the zero outcome of the coin flip), then the privacy cost is $0$. Otherwise, we need to analyze the privacy cost for releasing the sum of embedding vectors perturbed by the Gaussian noise. To do that, we need to first analyze the sensitivity of the sum of embedding vectors.

The worst-case $\mathcal{L}_2$ sensitivity of a bucket's embedding sum is the largest
embedding norm, which is unbounded; adding noise
scaled to it would destroy all utility. We instead use a \emph{probabilistic}
bound whose failure probability is paid in the final $\delta$.

\begin{lem}[High-probability Sensitivity]\label{lem:high-prob-sens}
Consider any bucket in the $k$-dimensional projection space, with bucket edge length $L$. It holds that the sensitivity for the sum of embedding vectors within the bucket is bounded by $\Delta_{dist}$, with probability 
\begin{align}
    \Pr[succeed]&\ge 1-\Bigg(f\Big(\frac{\sqrt{k} L}{\Delta_{\textit{dist}}}\Big)\Bigg)^k,\\
    \text{ with }\textstyle f(\nu)&:=\,\text{erf}\big(\frac{\nu}{\sqrt{2}}\big) - \sqrt{\frac{2}{\pi}} \frac{1}{\nu}(1-e^{-\nu^2/2})\label{eq:def-f-centralized}.
\end{align}
\end{lem}
We \emph{choose} the sensitivity
bound $\Delta_{\mathrm{dist}}$ so that the coupled pair
violates it has a small, negligible probability. Conditioned on its complement, Theorem~\ref{thm:main-gaussian} applies with sensitivity
$\Delta_{\mathrm{dist}}$. Our default experiment setting with $k=20, L=2r/\sqrt{k}$, and $\Delta_{\mathrm{dist}} = 2r$ (for any $r>0$) gives
$8.9\times10^{-11}$, a negligible share of the overall $\delta=10^{-6}$.

\begin{proof}[Proof sketch of Lemma~\ref{lem:high-prob-sens}]
Consider neighboring inputs $X' = X \cup \{x\}$ and condition on the added point
being assigned to bucket $h$. The sampled contributions to $h$
under $X$ and $X'$ can be coupled so that they differ in a single replaced point
$x'$ that is assigned to the \emph{same bucket} $h$. The sum over bucket $h$
therefore changes by at most $\lVert x - x' \rVert_2$ for two points $x, x'$ that
collide in all $k$ projected coordinates, and it suffices to bound the distance
of same-bucket collisions. Let $D = \lVert x - x' \rVert_2$. Each projection coordinate is
$y_j = \langle x, Z_j \rangle / \sqrt{k}$ with independent
$Z_j \sim \mathcal{N}(0, \mathbf{I}_d)$, so the coordinate-wise gap
$g_j = y_j - y'_j$ is i.i.d.\ $\mathcal{N}(0, D^2/k)$. Buckets have edge length
$L$ per coordinate with an independent uniform random offset, so two points with
gap $g_j$ share a cell in coordinate $j$ with probability
$\max(0,\, 1 - |g_j|/L)$. Taking the expectation over $g_j$ gives the
per-coordinate collision probability and
independence across the $k$ coordinates yields the statement. 
\end{proof}

\begin{theorem}[Gaussian mechanism]\label{thm:main-gaussian}
Let $\Delta_{\textit{dist}}$ be the $\mathcal{L}_2$ sensitivity for the vector sum and $\sigma$ be the parameter of the Gaussian noise. The centroids $\{\mu_h\}$ in Alg.~\ref{alg:centralized} satisfies $\rho_{agg}$-CDP, and $(\varepsilon_{agg},\delta_{agg})$-DP with \begin{align}
\rho_{agg}=\frac{\Delta_{\textit{dist}}^2}{2\sigma^2}, \textrm{ and }  
 \delta_{agg} = \inf_{\alpha \in (1,\infty)} \frac{e^{(\alpha-1)(\alpha\rho_{agg}-\varepsilon_{agg})}}{\alpha - 1} \cdot \left( 1 - \frac{1}{\alpha} \right)^{\alpha}.\end{align}
\end{theorem}

The overall DP guarantee comes from linear composition~\cite{dpbook}, and accounting for failure probability of sensitivity bound into $\delta$. 

\begin{theorem}[Overall DP guarantee]\label{thm:dp-final-fre-e2t} With $\mathcal{L}_2$ sensitivity $\Delta_{\textit{dist}}$, $f$ defined as in Eq.~\ref{eq:def-f-centralized}, and the parameters $\varepsilon_{fre},\varepsilon_{agg},\delta_{fre},\delta_{agg}$ computed as in Lemma~\ref{lem:s-and-t-dp} and Theorem~\ref{thm:main-gaussian}, Algorithm~\ref{alg:all} satisfies $(\varepsilon,\delta)$-DP for
\begin{align}
    \varepsilon = \varepsilon_{fre}+\varepsilon_{agg}\,, \textrm{and } 
\delta=\delta_{fre}+\delta_{agg}+\Big(f\big(\frac{\sqrt{k} L}{\Delta_{\textit{dist}}}\big)\Big)^k.
\end{align}
\end{theorem}

\subsection{Semantic Support Protection}\label{sec:semantic-protection}
For any infrequent user embedding $x^*$, there will not be any DP centroid $\widehat \mu$ such that $\widehat \mu$ is close to $x^*$ in terms of $\mathcal{L}_2$ distance. 

\begin{theorem}[SSP by Alg.~\ref{alg:centralized}]\label{thm:proof-semantic-all}
For any $x^*\in X$ such that there are at most $t$ data points from $X$ that has distance at most $r$ to $x^*$ in $X$ and any DP centroid $\widehat \mu$ with $h$ items subsampled into the cluster, we have $\|\widehat \mu - x^*\|_2 < \psi r$ with probability 
\begin{align}
\delta^* \le
\binom{h}{2}\Big(f\big(\frac{\sqrt{k}L}{r}\big)\Big)^k + e^{-0.38p_s t} + F_{\chi^2_d(\Lambda)}\Big(\frac{\psi^2 r^2h^2}{\sigma^2}\Big),
\end{align}
where $\Lambda=\frac{\tau-6p_st}{2\tau}r$, and $F$ is the cumulative density function (cdf) of noncentral chi-squared distribution
\begin{align*}
F_{\chi^2_d(\Lambda)}(a)
= e^{-\Lambda/2}\sum_{j=0}^{\infty}
\frac{(\Lambda/2)^j}{j!}\,
\frac{\gamma\!\left(\tfrac{d}{2}+j,\ \tfrac{a}{2}\right)}
{\Gamma\!\left(\tfrac{d}{2}+j\right)}\,,\,
a\ge 0, \text{ with }\\
\gamma(s,a) = \int_{0}^{a} t^{\,s-1} e^{-t}\,dt,
\Gamma(s) = \int_{0}^{\infty} t^{\,s-1} e^{-t}\,dt, s>0.
\end{align*}
\end{theorem}
Namely, Algorithm~\ref{alg:centralized} satisfies $(r,t,\psi r,\delta^*)$ semantic support protection for the DP \textit{centroids}. This $\delta^*$ quickly converges to $0$ as $k$, and $t$ increases, when $\sigma$, $\psi$, and $p_s$ are moderate and $L\approx \frac{r}{\sqrt{k}}$ times some constant. We include text-level evaluation in Section~\ref{sec:exp}, where additional error may come from the vector-to-text inversion process.

We sketch the proof (detailed in Appendix~\ref{app:proof-of-semantic-support}). We condition on the high-probability good event ``Within any bucket, there are at most $2p_st$ data points that are $r$ close to $x^*$ \textit{and} every two data points have distance at most $r$''. This can be proved using the fact that we cluster data points based on the $\mathcal{L}_2$ distance together with a concentration argument~\cite{MitzenmacherUpfal2005}. Next, we can argue that for any DP centroid in Alg.~\ref{alg:centralized}, the corresponding cluster has many data points that are far from $x^*$ and these ``far data points'' are close to each other. Based on this, we can argue that the average of all data points' embedding vectors in this cluster is far from $x^*$, since the average is dominated by the majority data points that are far from $x^*$. Finally, injecting random DP noise to the centroid will not decrease its distance to $x^*$ by a significant margin. Note that we do not claim that existing solutions such as~\cite{yueetal2023synthetic,xie2024differentially} do not achieve similar guarantees. Nevertheless, working directly on the embedding space with explicit truncation for infrequent texts has made the proof easier.\looseness=-1 

\section{Distributed Implementation}\label{sec:sol}

Algorithm~\ref{alg:all} implements Algorithm~\ref{alg:centralized} in the distributed environment, with threat model specified as in Section~\ref{sec:prob}. We refer to Figure~\ref{fig:illustration} in Appendix~\ref{app:alg-detail} for an illustration.

The implementation mainly makes use of three existing distributed protocols: (i) Verifiable oblivious pseudorandom function (VOPRF)~\citep{star}; (ii) distributed discrete Gaussian with secure aggregation (SecAgg)~\cite{secagg,DiscreteGaussian} for differential privacy; and (iii) a standard shuffler-style channel that strips per-user identifiers from messages~\cite{nebula,prochlo,
  star,dpshuff}. We treat each of these as a black-box with the 
security guarantees inherited from the original papers.

\begin{algorithm*}[t]
\DontPrintSemicolon
\KwIn{$N$ users, each one of which holds a private textual data point $v_i$; public embedding function $E$; embedding dimension $d$; projected dimension $k$; public Gaussian matrix $\mathbf{Z}\in\mathbb{R}^{d\times k}$; edge length $L$; public offsets $\{o_j\}_{j=1}^k$; tagging server \Stag; synthesis server \Ssyn; subsampling rate $p_s$; truncated shifted Laplace distribution $\text{TSDLap}(\lambda, \gamma)$; heavy-hitter threshold $\tau$; Discrete Gaussian noise std $\sigma$ and quantization granularity $\beta$.}
\KwOut{Synthesized texts.}
\For{User $i = 1, \ldots, N$}{
% Compute embedding vector $x_i\gets E(v_i)$ \tcp*{done  \textit{locally}}
% Compute projected embedding vector $y_i\gets \frac{1}{\sqrt{k}} x_i\mathbf{Z}$ \tcp*{done \textit{locally}}
Compute the $k$-dimensional bucket index $b_i$ for $v_i$ using \textit{Local-Preprocess}($v_i$) \tcp*{done locally}
Obtain $h_i$ using \textit{User-TaggingServer}($b_i$) \tcp*{interacting with the Tagging Server}
Sample $\text{coin}_i\in\{0,1\}$ with $p_s$ probability of $1$ \tcp*{done locally}
\If{$\text{coin}_i=1$}{
Submit index $h_i$ to the \textit{synthesis Server}\;
}
}
$\textsf{Dummy} =$ \textit{Dummy-Data-Creation}($\tau$, $\text{TSDLap}(\lambda, \gamma)$) \;
\textit{Synthesis server} receives
$\textsf{ReceivedData}=\textsf{sampled} \{h_i\} \cup \textsf{Dummy}$\;
\textit{Synthesis server}  filters out indices with counts smaller than $\tau$, obtaining index set $\mathcal{H}$ and corresponding counts $\mathcal{C}$\;
\textit{Synthesis server} broadcasts $\mathcal{H}$ and $\mathcal{C}$ to all users\;
\For{User $i = 1, \ldots, N$}{
\If{$\text{coin}_i=1$}{
\If{$h_i\in \mathcal{H}$}{
Obtain $\widetilde x_i = \textit{Discrete-Gaussian}(x_i,\sigma,\beta)$ \tcp*{Locally perturb the embedding vector}
Compute two secret shares $\widetilde x_i[1]$ and $\widetilde x_i[2]$ for $\widetilde x_i$ and
send to \textit{tagging server} and \textit{synthesis server}, respectively\; 
}
}
}
\For{$h\in \mathcal{H}$}{
\textit{Tagging server} and \textit{synthesis server} reconstruct $\widehat x^{(h)} = \frac{1}{\beta}\sum_{\widetilde x_i \text{ is with index } h} \widetilde x_i$ \;
% \tcp*{DP is guaranteed by the summation of independent local noises}
Synthesis server computes centroid $\widehat x^{(h)} \gets \frac{1}{c^{(h)}}\cdot \widehat x^{(h)} $ \tcp*{$c^{(h)}$ is the count of index $h$}
 Generate text $\widetilde{v}_h \gets \text{Invert}(\widehat{x}_h)$\;
 \textbf{Output} $\widetilde{v}_h$\;
}
\caption{Differentially text synthesis (distributed version of Fre-E2T)}
\label{alg:all}
\end{algorithm*}

\vspace{0pt}
\noindent\textbf{Step 1: Local pre-processing.} The local data pre-processing step is identical to the centralized setting, assuming that all users have access to the same embedding function $E$ and the same public randomness. In the end of this step, each user $i$ obtains a $k$-dimensional bucket index $b_i$ for its embedding vector $x_i$ (line 2 of Algorithm~\ref{alg:all}). 

\looseness=-1
   
\vspace{0pt}
\noindent\textbf{Step 2: Distributed heavy hitter estimation.} There is no trusted curator to directly process the bucket indices for finding heavy hitters. Instead, each user $i$ obliviously communicates with the tagging server to encrypt index $b_i$ into a random string $h_i$ (line 3 of Algorithm~\ref{alg:all}). User $i$ then flips a coin to decide whether to submit the encrypted string $h_i$ to the synthesis server (lines 4--6 of Algorithm~\ref{alg:all}). The synthesis server collects the encrypted indices from the users and computes a heavy-hitter histogram to identify frequent indices (lines 7--9 of Algorithm~\ref{alg:all}). 

\vspace{0pt}
\noindent\textbf{Step 3: Aggregation of frequent text embeddings.} The synthesis server broadcasts the frequent encrypted indices to all users (line 10 of Algorithm~\ref{alg:all}). Every user who shares the same encrypted index and has submitted the index in the last phase proceeds to securely aggregate their original data embedding vectors in the $d$-dimensional space with additive DP noise (lines 11-15 of Algorithm~\ref{alg:all}), obtaining a DP average for embedding vectors (i.e., a DP centroid for the cluster). For this phase, we need the two non-colluding servers, tagging server \Stag and synthesis server \Ssyn. The two servers can then combine their secret shares to reconstruct the noisy private centroid for each heavy hitting index (lines 16--19 of Algorithm~\ref{alg:all}). Note that only users who already submitted the encrypted index in Step~2 respond to
the Step~3 broadcast for that index. The synthesis server therefore
observes that ``some set of $\ge\tau$ anonymized endpoints share encrypted
index $h_i$'', which is exactly the heavy-hitter event whose existence is
already covered by the DP guarantee in
Lemma~\ref{lem:s-and-t-dp}. \looseness=-1

\vspace{0pt}
\noindent\textbf{Step 4: Offline post-processing.} Finally, the synthesis server uses each perturbed centroid to synthesize new textual data by inverting it back to a token sequence using the Vec2Text algorithm~\cite{morris2023language}. This is post-processing and does not cost additional privacy budget. 

\vspace{0pt}
\noindent\textbf{Paraphrasing.} In addition, one can also use open-sourced language models such as GPT-2~\cite{radford2019language} to paraphrase the generated texts to improve the diversity of the generated texts---e.g., $10$ times for each text inverted from a centroid. Inversion and paraphrasing steps require no further involvement from the users and can be done completely offline on the server side. In addition, assuming that the language model used for paraphrasing is not trained on the user texts, this paraphrasing step also incurs no additional privacy cost. 

Step 1 and 4 are done completely on the user and server side, respectively. Next, we discuss the distributed Step 2 and 3 in details. \looseness=-1

\subsection{Distributed Heavy Hitter Estimation}\label{sec:hh}

The idea is similar to the recent Nebula system~\citep{nebula} for distributed heavy hitter estimation. Each user $i$ first communicates obliviously with the tagging server to obtain the encrypted index $h_i$ for its $k$-dimensional bucket index $b_i$ (line 2 in Algorithm~\ref{alg:all}) while ensuring that the tagging server learns nothing about users’ bucket assignments. A Verifiable Oblivious
PseudoRandom Function~\citep{star} is applied to hash the $k$-dimensional bucket indices so that users in the same bucket obtain the same encrypted index without the tagging server noticing and no user-to-user communications (Algorithm~\ref{alg:hash} in Appendix~\ref{app:alg-detail}). 

Next, user $i$ submits $h_i$ with probability $p_s$ to the synthesis server (lines 4-6 in Algorithm~\ref{alg:all}), who then computes the aggregate counts and filters out all the indices that have not accumulated to $\tau$ submissions from the users (lines 8--10 in Algorithm~\ref{alg:all}). This truncated histogram satisfies DP.

In addition to that, the non-frequent (encrypted) indices should also receive proper privacy protection, since they are observed by server $S_{syn}$. To that end, we follow~\citep{nebula} and introduce a dummy user who can send additional dummy indices to the server such that the non-frequent indices will also receive DP protection while preserving the correctness of the aggregation result (line 7 in Algorithm~\ref{alg:all}). Roughly speaking, if the number of dummy indices follows the truncated shifted discrete Laplace distribution (TSDLap), then impact of an encrypted index is negligible compared with the dummy submission, ensuring DP.  We refer to Definition~\ref{def:tsd} and Algorithm~\ref{alg:create-dummy} in Appendix~\ref{app:alg-detail}  for more details.

\subsection{Distributed Centroid Computation}\label{sec:agg}

For each encrypted frequent index $h$, each user $i$ who has submitted its index to the aggregate server in Step 2 examines $h_i$: if it matches $h$, then the user obtain a perturbed version for the embedding vector $x_i$, written as $\widetilde x_i =  \textit{Discrete-Gaussian}(x_i,\sigma,\beta)$, where $\sigma$ is computed based on the target privacy budget $\varepsilon$ and $\delta$ and the quantization factor $\beta$, using the distributed discrete Gaussian mechanism~\cite{DiscreteGaussian}. Practical distributed DP protocols employ discrete noises~\cite{dgcanonne,DiscreteGaussian,sk} to reliably enforce privacy guarantees. However, we note that switching to discrete noises does not change the privacy--utility trade-off as long as the discretization granularity is not too coarse, according to~\cite{DiscreteGaussian,sk}, which is why we simulate with the continuous Gaussian noise in our experiments. We defer more technical details to Section~\ref{sec:privacy-ddg} and Appendix~\ref{app:alg-detail}.

Next, the perturbed data points with the same frequent index $h$ are aggregated securely, then divided by the occurrence $c^{(h)}$ and the quantization parameter $\beta$. Here, the DP guarantee is collectively contributed by the local discrete Gaussian noises generated by users with the same index $h$. To perform the secure aggregation, we reuse the tagging server and synthesis server from heavy hitter estimation: each user computes two secret shares for her/his perturbed embedding vector and sends them to the servers. The secret shares are tagged by the encrypted index $h_i$ so that the servers can later combine them to reconstruct the overall sum, as is done in federated analytics~\cite{corrigan2017prio,samplablefa}. One may also consider using the implementations from~\cite{secagg,fu2024benchmarking}. Note that we cannot use the anonymous channel to replace distributed DP protocol, since the shuffle-model amplification for vector sums attains meaningful guarantees only in restricted $\varepsilon$ regimes~\cite{dpshuff}; while secret-shared aggregation reconstructs the exact noisy sum at any $\varepsilon$. 

\vspace{0pt}
\noindent\textbf{Deployment costs.}
Recall that we write $N$ for the number of users, $d$ for the embedding dimension, $k$ for the
projection dimension, $p_s$ for the subsampling rate, $\tau$ for the heavy-hitter
threshold, $H$ for the set of released buckets, and $S=\sum_{h\in H}c(h)$ for the
number of submissions that participate in the centroid phase.
On the user side, it first embeds its own text once and projects
it in $O(dk)$ time offline. The distributed protocol is \emph{one-shot}: the heavy-hitter
phase costs each user $O(1)$, and a user whose bucket is
released spends a further $O(d\log d)$ for the rotation and $O(d)$ for
quantization, discrete-Gaussian noise and secret sharing. Per-user communication
is $O(1)$ upstream in the heavy-hitter phase and $O(d)$ upstream if
participating in the centroid phase, and $O(|H|)$ downstream to receive the released bucket set. For the server side, the tagging server\ performs $O(N)$ VOPRF
evaluations and $O(Sd)$ share accumulation, while the synthesis server\ performs
$O(p_sN)$ tag counting, $O(Sd)$ accumulation and
$O(|H|\,d\log d)$ reconstruction. 

\vspace{0pt}
\noindent\textbf{Party roles and observations.} User~$i$ performs all data-dependent computation on its own device: it embeds text
$v_i$, projects it, derives the bucket index $b_i$, obtains the tag $h_i$, draws its participation coin and its noise, and forms its two shares for DP centroid computation. It learns about other users only through the public release: the broadcast heavy-hitter histogram and the published centroids, both of which are differentially private results.

The tagging server \Stag issues the bucket tags obliviously, generates the dummy tags that pad the
sub-threshold histogram (Algorithm~\ref{alg:create-dummy}), and contributes one of the two
additive shares of every perturbed embedding, sending \Ssyn{} the
per-bucket aggregate share. \looseness=-1

The synthesis server \Ssyn collects the sampled tag reports, counts them, discards tags below the
threshold~$\tau$, broadcasts heavy-hitting tags with their sampled counts, combines the two secret shares into the noisy per-bucket sum, and inverts the noisy centroids into synthetic text offline. Our privacy guarantee applies to the  whole multiset of tags with their
multiplicities and the DP centroids; therefore applies to every party. 

\subsection{Privacy in Distributed Settings}\label{sec:privacy-ddg}
The privacy guarantees of the distributed implementation are largely similar to Algorithm~\ref{alg:centralized}. We highlight some differences.\looseness=-1

\vspace{0pt}
\noindent\textbf{Non-released indices in distributed heavy hitters.} For each index that has accumulated $<\tau$ user submissions, the synthesis server observes the encrypted version of its index, and the noisy count. To argue the privacy for such statistics, we note that the index itself carries no information, as it looks random to any computationally bounded adversary~\cite{computationaldp}. So, the problem becomes analyzing the privacy cost for releasing an anonymized histogram (sometimes referred to as the count-of-counts~\cite{census2010,hcofc}), which reports the number of indices that have a count of $Q$, for each $Q=1,2,...,\tau-1$. According to~\citep{Desfontaines2022Differentially,sparsetwoserver,nebula}, injecting $\textsf{TSDLap}(\lambda,\gamma)$ into the number of indices of each count satisfies $(\frac{2}{\lambda},\frac{1}{2}\exp(-\frac{\gamma-2}{\lambda}))$-DP, which is exactly why we need to create dummy indices as in Alg.~\ref{alg:create-dummy}. The subsampling ratio of $p_s$ amplifies the privacy guarantee to $(\varepsilon_{unre},\delta_{unre})$-DP with
\begin{align}
\textstyle
    \varepsilon_{unre} = \ln\Big(1+p_s \big(e^{\frac{2}{\lambda}}-1\big)\Big),\textrm{ and }\delta_{unre} = \frac{p_s}{2}e^{-\frac{\gamma-2}{\lambda}}\,.\label{eq:dp-unre-delta}
\end{align}
We can then obtain the following  guarantee for all indices. 

\begin{lem}[DP heavy hitter estimation]\label{lem:hh-all}
For any adversary who observes the broadcast heavy hitter histogram and the synthesis server who observes the counts of all indices (including dummies), Algorithm~\ref{alg:all} satisfies $\big(\max(\varepsilon_{fre},\varepsilon_{unre}),\max(\delta_{fre},\delta_{unre})\big)$-DP, where $\varepsilon_{fre}$ and $\delta_{fre}$ are from Lemma~\ref{lem:s-and-t-dp}; $\varepsilon_{unre}$ and $\delta_{unre}$ are from Eq.~\ref{eq:dp-unre-delta}.
\end{lem}

\vspace{0pt}
\noindent\textbf{Computing Centroids with Discrete Gaussian.} The data pre-processing is not influenced by the distributed protocol. Hence, we can use the same $\mathcal{L}_2$ sensitivity $\Delta_{\textit{dist}}$, with a high probability. Conditioned on this, we can apply the distributed discrete Gaussian mechanism (DDG)~\cite{DiscreteGaussian}.  More concretely, the user first quantizes the embedding vector $x_i\in\mathbb{R}^d$ with parameter $\beta$ to an integer-valued vector in $\mathbb{Z}^d$ by multiplying the real vector by $1/\beta$ and then randomly rounding each coordinate to the nearest integer. We note that while some biases will appear due to rounding, they can be offset to a negligibly small influence on the final result if $\beta$ is small enough, according to~\cite{DiscreteGaussian}. Then, for each coordinate of the integer-valued vector, the user perturbs it with a discrete Gaussian noise $\mathcal{N}_{\mathbb{Z}}(0,\sigma^2)$ independently, computes two secret shares for the perturbed vector, and then sends them to the synthesis server and the perturbation server along with the encrypted index, using the framework suggested in~\cite{corrigan2017prio,samplablefa}. During secure aggregation, the discrete Gaussian noises contributed by independent users will aggregate into a larger noise (of standard deviation at least $\sqrt{\tau}\sigma$) to amplify the DP guarantee. After obtaining the aggregated result, the synthesis server will then multiply it by $\beta$ to obtain an estimation for the original vector sum. We give the privacy guarantee of DDG as follows for completeness. Details are in Appendix~\ref{app:alg-detail}. \looseness=-1

\begin{theorem}[Distributed discrete Gaussian~\cite{DiscreteGaussian}]\label{thm:main-dgg}
Let $\beta>0$ be the quantization parameter and $\sigma$ be the parameter of the local discrete Gaussian noise. Under the same assumption as Theorem~\ref{thm:main-gaussian}, define
\begin{align}
\rho_{agg} &:= \min \Big(
  \frac{\Delta_2^2}{2\tau \sigma^2} +  \kappa d, \;
  \frac{1}{2}\big(\frac{\Delta_2}{\sqrt{\tau} \sigma} + \kappa\sqrt{d}\big)^2
\Big),
\end{align}
where $\Delta_2^2=O(\Delta_\textit{dist}^2)$, and $\kappa = 10 \cdot \sum_{j=1}^{\tau-1} \exp\big(- \frac{2\pi^2 \sigma^2}{\beta^2} \cdot \frac{j}{j+1}\big)$ is some constant; see Theorem~\ref{thm:main-dgg-formal} for more details. The perturbed sum of embedding vectors from Alg.~\ref{alg:all} satisfies $\rho_{agg}$-CDP, and $(\varepsilon_{agg},\delta_{agg})$-DP with \(
    \delta_{agg} = \inf_{\alpha \in (1,\infty)} \frac{\exp((\alpha-1)(\alpha\rho_{agg}-\varepsilon_{agg}))}{\alpha - 1} \cdot \left( 1 - \frac{1}{\alpha} \right)^{\alpha}.
\)
\end{theorem}

\vspace{0pt}
\noindent\textbf{Overall DP guarantee.} The overall DP guarantee comes from linearly composing the parameters from Lemma~\ref{lem:hh-all} and Theorem~\ref{thm:main-dgg}, and accounting for failure probability of sensitivity bound into the final $\delta$. Inverting centroids to texts costs no privacy. The only difference is that composition of parameters only occurs on the data points that appear in frequent buckets. Those data points that are not frequent do not participate in the release of DP sum of embedding vectors; hence, they compose in parallel~\cite{dpbook}. 

\begin{theorem}[Overall DP guarantee]\label{thm:dp-final} With the parameters specified as in Lemma~\ref{lem:s-and-t-dp}, Lemma~\ref{lem:high-prob-sens}, Lemma~\ref{lem:hh-all}, and Theorem~\ref{thm:main-dgg},  and $f$ defined as in Eq.~\ref{eq:def-f-centralized}, Algorithm~\ref{alg:all} satisfies $(\varepsilon,\delta)$-DP for any adversary with
\begin{align}
    \varepsilon &= \max(\varepsilon_{fre}+\varepsilon_{agg},\varepsilon_{unre})\,, \textrm{and}   \\
\delta&=\max\Big(\delta_{fre}+\delta_{agg}+\Big(f\big(\frac{\sqrt{k} L}{\Delta_{\textit{dist}}}\big)\Big)^k,\delta_{unre}\Big). 
\end{align}
\end{theorem}

\section{Experiments}\label{sec:exp}
We demonstrate the utility of our algorithm and compare it with the state of the art ~\cite{xie2024differentially}. Our goal is to show that our \name achieves utility comparable to the strong centralized baseline under the same privacy constraints while being more distributed-friendly and offering additional protection for infrequent texts.

\vspace{0pt}
\noindent \textbf{Datasets.}
Our experiments include four datasets containing user queries to search engines/LLMs: {TaylorAI}~\cite{tayloraiuserqueriesdataset}, {Instructions-2M}~\citep{morris2023language}, {LMSYS Chat}~\citep{zheng2023lmsyschat1m}, and {Yelp}~\cite{zhang2015character}. See Appendix~\ref{app:exp} for details.\looseness=-1

\vspace{0pt}
\noindent \textbf{Infrequent and frequent texts.} Recall Section~\ref{sec:semantic-support}. We enforce additional privacy protection for infrequent user texts that have fewer than $t$ neighbors in its $r$-neighborhood in the embedding space. Conversely, we say a text is frequent if it has more than $t$ neighbors in the $r$-neighborhood. We take $r=0.5$ and $t=100$, unless otherwise stated.

\vspace{0pt} \noindent \textbf{Baseline and implementation.} We focus on the state-of-the-art baseline for DP text synthesis from the centralized setting, i.e., Aug-PE~\cite{xie2024differentially}, and use it as the reference point for evaluating the utility of our distributed \name. We omit other baselines as they achieve similar performance as~\cite{xie2024differentially}, while some of them relying on additional gradient access to the LM~\cite{yueetal2023synthetic,flemingsetal2024differentially,hou2025private}. Since we assume no label information is available before hand and some datasets such as ~\cite{tayloraiuserqueriesdataset,morris2023language} do not have such information, we remove such information (if there is any) from the prompts for a fair comparison. We follow their setup and use Aug-PE's published
hyperparameters. We use GPT-2~\cite{radford2019language} as the generator API for~\cite{xie2024differentially} and \textsf{gtr-t5-base}~\citep{ni2022large} as the embedding model. \textsf{gtr-t5-base} maps each text input into a 768-dimensional vector. \looseness=-1

For our implementation, we simulate our algorithm on a single machine with $4\times$ A100 GPUs. We use the same embedding model. For the projected dimension, we focus on $k=20$ unless otherwise stated. We set the heavy hitter threshold $\tau=p_s t$, where $p_s$ is the subsampling rate controlling the privacy budget for the heavy hitter release. We vary $p_s\in\{0.3,0.5,0.6\}$. Next, the Gaussian noise parameter $\sigma$ for centroid computation is determined by the remaining privacy budget. Throughout we fix the overall $\delta=10^{-6}$ and vary the overall $\varepsilon\in\{4,8,16,\infty\}$. For $\varepsilon=\infty$, we run our Fre-E2T without any additional DP noise or client subsampling. We leverage the open-sourced \textsf{vec2text}~\cite{morris-etal-2023-text} to invert embeddings back to texts. This model contains 235M parameters. By default, the server uses \textsf{vec2text} to synthesize one text for each obtained DP centroid of embedding vectors. For our Fre-E2T, we also include a variant that uses GPT-2 to obtain more variations of each synthetic text generated from the released embedding vector, ensuring a fair comparison with Aug-PE. GPT-2 never touches the user data without affecting DP (see Section~\ref{sec:sol}). To ensure a fair comparison, both Aug-PE and our Fre-E2T obtain $3000$ texts in total.  

\vspace{0pt} 
\noindent\textbf{Utility metrics.}
Given a synthetic dataset $\widetilde V$ and the original private dataset $V$, and constants $r$ and $t$, we measure utility using three metrics that largely follow prior work~\cite{yueetal2023synthetic,xie2024differentially}: precision (higher is better), recall (higher is better), and average $\mathcal{L}_2$ distance (lower is better). We also report F1 score, the harmonic mean of precision and recall. The main here difference is that we compute these metrics on semantically frequent texts, i.e., $V\setminus V_{\text{infre}}(r,t)$, since Definition~\ref{def:semantic-protect} explicitly restricts outputs near semantically infrequent texts. 

\begin{align}\label{eq:prec}
\hspace{-0.5em}\text{Prec}:=\frac{\sum_{\widetilde v\in \widetilde V}\mathbf{1}\big(\min_{v\in V\setminus V_{\text{infre}}(r,t) }\|E(\widetilde v)-E(v)\|_2\leq \alpha r\big)}{|\widetilde V|}.
\end{align}

\begin{align}\label{eq:rec}
\hspace{-0.5em}\text{Rec}:&= \frac{\sum_{v\in V\setminus V_{\text{infre}}(r,t)}\mathbf{1}\big(\min_{\widetilde v\in \widetilde V}\|E(\widetilde v)-E(v)\|_2\leq \alpha r\big)}{|V\setminus V_{\text{infre}}(r,t)|}.
\end{align}

\begin{align}\label{eq:l2}
\hspace{-0.5em}\text{Dist}:&=    \frac{\sum_{\widetilde v\in \widetilde V}\min_{v\in V\setminus V_{\text{infre}}(r,t)}\|E(\widetilde v)-E(v)\|_2}{|\widetilde V|}.
\end{align}

Here, a larger $\alpha$ makes the neighborhood condition in Precision/Recall easier to satisfy. We fix $\alpha=1.5$ throughout unless otherwise stated, giving a threshold $\alpha r=0.75$.

\subsection{Main Results}

\begin{table*}[t]
    \centering
    \small
    \caption{Precision, recall, F1 score, and $\mathcal{L}_2$ distance (lower is better) across
    datasets for $\varepsilon\in\{4,8,16,\infty\}$. $\varepsilon=\infty$
    denotes the non-DP setting. GPT-2 is used for paraphrasing by Fre-E2T
    and for generation and variation by Aug-PE.}
    \label{tab:combined_dp_results}

    \setlength{\tabcolsep}{2pt}
    \renewcommand{\arraystretch}{0.95}

    \begin{tabular*}{\textwidth}{
        @{\extracolsep{\fill}}
        cc *{3}{cccc}
        @{}
    }
    \toprule
    \multirow[c]{2}{*}{Dataset}
    & \multirow[c]{2}{*}{$\varepsilon$}
    & \multicolumn{4}{c}{\shortstack{Fre-E2T\\(no GPT-2)}}
    & \multicolumn{4}{c}{\shortstack{Fre-E2T\\(GPT-2)}}
    & \multicolumn{4}{c}{\shortstack{Aug-PE\\(GPT-2)}} \\
    \cmidrule(lr){3-6}\cmidrule(lr){7-10}\cmidrule(lr){11-14}
    &
    & Precision & Recall & F1 & $\mathcal{L}_2$
    & Precision & Recall & F1 & $\mathcal{L}_2$
    & Precision & Recall & F1 & $\mathcal{L}_2$ \\
    \midrule

    \multirow[c]{4}{*}{TaylorAI}
    & 4
    & $0.31$ & $0.92$ & $0.47$ & $0.82$
    & \textbf{1.00} & $0.98$ & \textbf{0.99} & $0.41$
    & $0.99$ & \textbf{0.99} & \textbf{0.99} & \textbf{0.40} \\
    & 8
    & $0.62$ & $0.95$ & $0.75$ & $0.74$
    & \textbf{1.00} & $0.98$ & \textbf{0.99} & $0.41$
    & $0.99$ & \textbf{0.99} & \textbf{0.99} & \textbf{0.40} \\
    & 16
    & $0.82$ & $0.97$ & $0.88$ & $0.68$
    & \textbf{1.00} & $0.98$ & \textbf{0.99} & $0.41$
    & $0.98$ & \textbf{0.99} & $0.98$ & \textbf{0.40} \\
    & $\infty$
    & $0.83$ & $0.97$ & $0.89$ & $0.67$
    & \textbf{1.00} & $0.98$ & \textbf{0.99} & \textbf{0.40}
    & $0.98$ & \textbf{0.99} & $0.98$ & \textbf{0.40} \\

    \midrule

    \multirow[c]{4}{*}{\shortstack{Instructions-\\2M}}
    & 4
    & $0.25$ & $0.32$ & $0.28$ & $0.92$
    & $0.97$ & \textbf{0.71} & \textbf{0.82} & \textbf{0.61}
    & \textbf{0.98} & $0.60$ & $0.75$ & $0.84$ \\
    & 8
    & $0.38$ & $0.34$ & $0.36$ & $0.80$
    & $0.97$ & \textbf{0.74} & \textbf{0.84} & \textbf{0.57}
    & \textbf{0.98} & $0.60$ & $0.75$ & $0.83$ \\
    & 16
    & $0.56$ & $0.37$ & $0.45$ & $0.71$
    & $0.98$ & \textbf{0.76} & \textbf{0.86} & \textbf{0.54}
    & \textbf{0.99} & $0.59$ & $0.74$ & $0.84$ \\
    & $\infty$
    & $0.69$ & $0.39$ & $0.50$ & $0.62$
    & \textbf{0.99} & \textbf{0.79} & \textbf{0.88} & \textbf{0.54}
    & \textbf{0.99} & $0.60$ & $0.75$ & $0.83$ \\

    \midrule

    \multirow[c]{4}{*}{LMSYS Chat}
    & 4
    & $0.31$ & $0.74$ & $0.43$ & $0.92$
    & \textbf{0.96} & \textbf{0.82} & \textbf{0.89} & \textbf{0.73}
    & $0.92$ & $0.80$ & $0.86$ & $0.86$ \\
    & 8
    & $0.52$ & $0.83$ & $0.64$ & $0.85$
    &\textbf{0.96} & \textbf{0.87} & \textbf{0.91} & \textbf{0.56}
    & $0.92$ & $0.80$ & $0.86$ & $0.91$ \\
    & 16
    & $0.75$ & $0.86$ & $0.80$ & $0.69$
    & \textbf{0.97} & \textbf{0.88} & \textbf{0.92} & \textbf{0.53}
    & $0.90$ & $0.81$ & $0.85$ & $0.89$ \\
    & $\infty$
    & $0.83$ & $0.84$ & $0.83$ & $0.67$
    & \textbf{0.97} & \textbf{0.86} & \textbf{0.91} & \textbf{0.57}
    & $0.91$ & $0.80$ & $0.85$ & $0.91$ \\

    \midrule

    \multirow[c]{4}{*}{Yelp}
    & 4
    & $0.30$ & $0.96$ & $0.46$ & $0.82$
    & \textbf{1.00} & \textbf{1.00} & \textbf{1.00} & \textbf{0.37}
    & $0.98$ & \textbf{1.00} & $0.99$ & \textbf{0.37} \\
    & 8
    & $0.65$ & \textbf{1.00} & $0.76$ & $0.72$
    & \textbf{1.00} & \textbf{1.00} & \textbf{1.00} & \textbf{0.37}
    & $0.94$ & \textbf{1.00} & $0.97$ & \textbf{0.37} \\
    & 16
    & $0.87$ & \textbf{1.00} & $0.93$ & $0.67$
    & \textbf{1.00} & \textbf{1.00} & \textbf{1.00} & $0.38$
    & $0.98$ & \textbf{1.00} & $0.99$ & \textbf{0.36} \\
    & $\infty$
    & $0.97$ & \textbf{1.00} & $0.98$ & $0.64$
    & \textbf{1.00} & \textbf{1.00} & \textbf{1.00} & \textbf{0.37}
    & $0.95$ & \textbf{1.00} & $0.97$ & \textbf{0.37} \\

    \bottomrule
    \end{tabular*}
\end{table*}

\noindent\textbf{Precision and recall.} We report the utility of our \name under different privacy budgets in Table~\ref{tab:combined_dp_results}. Overall, the precision, recall, and F1 score increases as the privacy constraint $\epsilon$ increases from $4$ to $\infty$. Namely, as the privacy constraint loosens, the synthetic texts achieve higher utility. This trend is consistent across all four datasets. On Yelp and TaylorAI, both Aug-PE~\cite{xie2024differentially} and our method achieve almost perfect precision and recall. On Instructions-2M and LMSYS Chat, our \name achieves {notable improvements}. We attribute such differences to the datasets: TaylorAI and Yelp have long sequences, leading to embedding vectors of smaller norms in general, hence easier to synthesize due to a smaller space. Other datasets contain shorter sequences and are in general more difficult to obtain closer synthetic texts.

Comparing the utility of \name (no GPT-2) and \name (GPT-2), we see that paraphrasing leads to higher utility. Utilizing stronger embedding models and corresponding \textsf{vec2text} reconstruction models will likely lead to further improvements. We leave that as a future work direction (discussed in Section~\ref{sec:design-choices}). 

\vspace{0pt}
\noindent\textbf{$\mathcal{L}_2$ distance.} The distance metric
of Eq.~\ref{eq:l2} is an absolute distance in embedding space. Here, we include two references using the same embedding model. The \emph{same-dataset} reference replaces the synthetic set by $3000$ texts drawn uniformly from the frequent subset and scores them against the remaining frequent texts, setting the lower bound. The \emph{unrelated} reference instead draws $3000$ texts uniformly from the other three datasets and scores them against the target corpus: demonstrating what the distance metric looks like when the released text is irrelevant with the data. On TaylorAI, Instructions-2M, LMSYS Chat, and Yelp, the \emph{same-dataset} and \emph{unrelated} references are $0.20$ and $1.37$, $0.12$ and $1.42$, $0.06$ and $1.29$, and $0.26$ and $1.51$, respectively. Both methods fall into the range set up by the two references on all four datasets.

On Instructions-2M and LMSYS Chat, Fre-E2T is closer to the real input texts than Aug-PE under every privacy budget. On TaylorAI and Yelp, the two are within 0.02 of each other and quite close to the lower bound. Aug-PE and Fre-E2T also differ in how they use the privacy budget. The distance of \name falls monotonically with $\varepsilon$ on all four datasets shown, whereas Aug-PE barely moves. Similar trends are observed in Table~\ref{tab:combined_dp_results}. We explain this next.

To further inspect the closeness between a synthetic text and a input private text, in Figure~\ref{fig:hist-l2}, we report histograms of the smallest $\mathcal{L}_2$ distance from a frequent text (in blue) and an infrequent text (in red) to a synthetic text, with Instructions-2M dataset as the input. For private texts that are frequent, the $\mathcal{L}_2$ distances to the nearest synthetic texts are smaller, compared with infrequent texts, which is aligned with our utility objective and the privacy framework of \textit{semantic support protection}---texts that are infrequent in the input should not find close neighbors in the synthetic output. We observe similar trends on other datasets and omit them here. 

\begin{figure}[t]
\centering
    \centering
    % --- Row 1, Plot 1 ---
    \begin{tikzpicture}[baseline]
        \begin{axis}[
          width=0.58\linewidth, height=3.0cm, % Changed from 0.55 to 0.48
          xlabel=Distance, ylabel=Density,
          ylabel near ticks,         xlabel near ticks,
          label style={font=\scriptsize},
          tick label style={font=\scriptsize},
          xmin=0, xmax=2.5,
          grid=major, grid style={dashed, gray!30},
          bar width=2pt,
          title={$\epsilon=2$},
          title style={font=\scriptsize}
        ]
        \addplot[ybar, fill=red, fill opacity=0.3, draw=red!70!black]
          table[col sep=space, header=true, x=distance, y=uncommon]
          {data/instructions2m_distance_histogram_eps2.dat};
        \addplot[ybar, fill=blue, fill opacity=0.3, draw=blue!70!black, bar shift=2pt]
          table[col sep=space, header=true, x=distance, y=common]
          {data/instructions2m_distance_histogram_eps2.dat};
        \end{axis}
    \end{tikzpicture}% <--- 
    % --- Row 1, Plot 2 ---
    \begin{tikzpicture}[baseline]
        \begin{axis}[
          width=0.58\linewidth, height=3.0cm, % Changed from 0.55 to 0.48
          ymax=8,
          xlabel=Distance, 
          ylabel near ticks,         xlabel near ticks,
          label style={font=\scriptsize},
          tick label style={font=\scriptsize},
          xmin=0, xmax=2.5,
          grid=major, grid style={dashed, gray!30},
          bar width=2pt,
          title={$\epsilon=4$},
          title style={font=\scriptsize},
          legend style={
        font=\scriptsize,
            at={(0.98,0.98)},
            anchor=north east,
            inner xsep=1pt,
            inner ysep=1pt
          }
        ]
        \addplot[ybar, fill=red, fill opacity=0.3, draw=red!70!black]
          table[col sep=space, header=true, x=distance, y=uncommon]
          {data/instructions2m_distance_histogram_eps4.dat};
        \addplot[ybar, fill=blue, fill opacity=0.3, draw=blue!70!black, bar shift=2pt]
          table[col sep=space, header=true, x=distance, y=common]
          {data/instructions2m_distance_histogram_eps4.dat};
    \legend{Infrequent, Frequent}
    \end{axis}
    \end{tikzpicture}
    
    \vspace{0.1cm} 
    
    % --- Row 2, Plot 3 ---
    \begin{tikzpicture}[baseline]
        \begin{axis}[
          width=0.58\linewidth, height=3.0cm, 
          ymax=8,
          xlabel=Distance, ylabel=Density,
          ylabel near ticks,         xlabel near ticks,
          label style={font=\scriptsize},
          tick label style={font=\scriptsize},
          label style={font=\scriptsize},
          tick label style={font=\scriptsize},
          xmin=0, xmax=2.5,
          grid=major, grid style={dashed, gray!30},
          bar width=2pt,
          title={$\epsilon=8$},
          title style={font=\scriptsize}
        ]
        \addplot[ybar, fill=red, fill opacity=0.3, draw=red!70!black]
          table[col sep=space, header=true, x=distance, y=uncommon]
          {data/instructions2m_distance_histogram_eps8.dat};
        \addplot[ybar, fill=blue, fill opacity=0.3, draw=blue!70!black, bar shift=2pt]
          table[col sep=space, header=true, x=distance, y=common]
          {data/instructions2m_distance_histogram_eps8.dat};
        \end{axis}
    \end{tikzpicture}% <--- 
    % --- Row 2, Plot 4 ---
    \begin{tikzpicture}[baseline]
        \begin{axis}[
          width=0.58\linewidth, height=3.0cm, % Changed from 0.55 to 0.48
          ymax=8,
          xlabel=Distance, 
          ylabel near ticks,         xlabel near ticks,
          xmin=0, xmax=2.5,
          grid=major, grid style={dashed, gray!30},
          label style={font=\scriptsize},
          tick label style={font=\scriptsize},
          legend style={
            font=\scriptsize,
            at={(0.98,0.98)},
            anchor=north east,
            inner xsep=1pt,
            inner ysep=1pt
          },
          bar width=2pt,
          title={$\epsilon=\infty$ (Non-DP)},
          title style={font=\scriptsize}
        ]
        \addplot[ybar, fill=red, fill opacity=0.3, draw=red!70!black]
          table[col sep=space, header=true, x=distance, y=uncommon]
          {data/instructions2m_distance_histogram_non_dp.dat};
        \addplot[ybar, fill=blue, fill opacity=0.3, draw=blue!70!black, bar shift=2pt]
          table[col sep=space, header=true, x=distance, y=common]
          {data/instructions2m_distance_histogram_non_dp.dat};
        % \legend{Infrequent, Frequent}
        \end{axis}
    \end{tikzpicture}
\caption{The smallest distance from a frequent/infrequent text (blue/red) to a synthetic text on Instructions-2M. 
}
\label{fig:hist-l2}
\Description{Four histograms compare distances to the nearest synthetic text for frequent and infrequent Instructions-2M texts at privacy budgets 2, 4, 8, and the non-private setting. Frequent texts generally have smaller distances than infrequent texts.}
\end{figure}
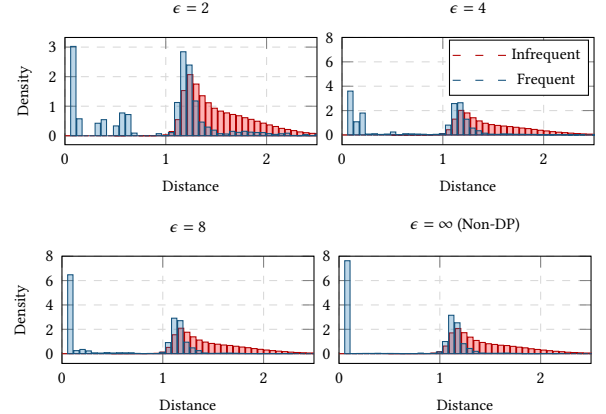

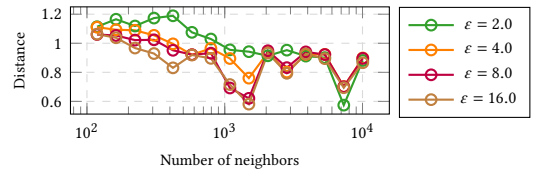
\begin{figure}[t]
\centering
\begin{tikzpicture}
    \begin{axis}[
        xmode=log,
        width=5.8cm, height=3.0cm,
        xlabel={Number of neighbors},
        ylabel={Distance},
        ylabel near ticks,
        label style={font=\scriptsize},
        tick label style={font=\scriptsize},
        legend pos=outer north east,
        legend style={font=\scriptsize},
        grid=major, grid style={dashed, gray!30},
        mark size=2pt,
        every axis plot/.append style={thick}
    ]
    % DP with eps=2.0
    \addplot[color=green, mark=o] table[x=x, y=y, col sep=comma] 
        {data/instructions2m_distance_vs_neighbors_eps2.csv};
    \addlegendentry{$\varepsilon=2.0$}
    % DP with eps=4.0
    \addplot[color=orange, mark=o] table[x=x, y=y, col sep=comma] 
        {data/instructions2m_distance_vs_neighbors_eps4.csv};
    \addlegendentry{$\varepsilon=4.0$}
    % DP with eps=8.0
    \addplot[color=purple, mark=o] table[x=x, y=y, col sep=comma] 
        {data/instructions2m_distance_vs_neighbors_eps8.csv};
    \addlegendentry{$\varepsilon=8.0$}
    % DP with eps=16.0
   \addplot[color=brown, mark=o] table[x=x, y=y, col sep=comma] 
       {data/instructions2m_distance_vs_neighbors_eps16.csv};
    \addlegendentry{$\varepsilon=16.0$}
    \end{axis}
    \end{tikzpicture}
\caption{Average distances to synthetic texts for frequent texts as a function of neighbor number on Instructions-2M. }
\label{fig:sim_vs_nn}
\Description{Line plots of distance to synthetic texts against the number of neighbors, with a logarithmic horizontal axis. Curves compare privacy budgets 2, 4, 8, and 16 on Instructions-2M. Distances generally decrease as the number of neighbors increases.}
\end{figure}

We also group the frequent texts from Instructions-2M dataset by their number of neighbors within distance $r=0.5$, and report the average $\mathcal{L}_2$ distance to the closest synthetic text within each frequent-text group in Figure~\ref{fig:sim_vs_nn}. The trend is similar: as the number of neighbors of a frequent text increases, the $\mathcal{L}_2$ distance to its closest synthetic text decreases. This aligns with our algorithm design: for each heavy hitting bucket, the count of items in general increases as the number of neighbors for a frequent text therein increases. As we inject DP noises to the sum of embedding vectors, the DP noise averaged by the bucket count decreases as the count increases, leading to smaller $\mathcal{L}_2$ distances. 

\vspace{0pt}
\noindent\textbf{Comparison with Aug-PE~\cite{xie2024differentially}}. Aug-PE refines a population of candidate texts by having each private text vote for its nearest candidate, releasing that vote histogram under DP, and asking a language model for variations of the candidates that receive more votes. When no label or category information is available, the histogram is the only signal connecting the candidate synthetic texts with the real inputs. It reports how many private texts lie near a candidate, but not which region of the input distribution that candidate should evolve towards. The variation step is therefore undirected and the population never separates into clusters matching the structure of the input data, so reducing the noise on the histogram only sharpens a signal that carries little directional information. That is what the flat trend of Aug-PE shows. The same insensitivity of text utility to the privacy budget is also observed in the original paper~\cite{xie2024differentially}. Fre-E2T supplies this missing structure: it clusters the frequent texts and generates each released text from one cluster centroid. This centroid is a soft label that tells what the texts therein is meant to represent in the embedding space. The privacy budget influences the accuracy of those centroids; hence, less noise translates directly into shorter distances and higher precision and recall. \looseness=-1

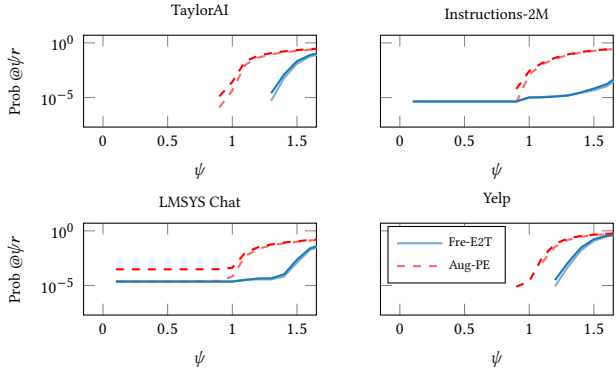
\begin{figure}[t!]
\centering
% Required in the preamble:
% \usepackage{pgfplots}

    \begin{tikzpicture}
    \begin{groupplot}[
        group style={
            group size=2 by 2,
            horizontal sep=0.1\columnwidth,
            vertical sep=0.15\columnwidth,
        },
        width=0.55\columnwidth,
        height=0.33\columnwidth,
        xmin=0,
        xmax=1.5,
        ymin=1e-7,
        ymax=1,
        ymode=log,
        xlabel={$\psi$},
        xlabel near ticks,
        ylabel near ticks,
        enlarge x limits=0.1,
        enlarge y limits=0.1,
        label style={font=\scriptsize},
        title style={font=\scriptsize,yshift=-3pt},
        tick label style={font=\scriptsize},
        unbounded coords=jump,
    ]

    % ----------------------------------------------------------
    % Top-left: TaylorAI
    % ----------------------------------------------------------
    \nextgroupplot[
        title={TaylorAI},
        ylabel={Prob @$\psi r$},
    ]

    \addplot+[
        color=blue!60,
        mark=none,
        thick
    ] table[
        col sep=comma,
        skip first n=1,
        x index=0,
        y index=2
    ]{data/semantic_support_taylor.csv};

    \addplot+[
        color=red!60,
        mark=none,
        dashed,
        thick
    ] table[
        col sep=comma,
        skip first n=1,
        x index=0,
        y index=12
    ]{data/semantic_support_taylor.csv};

    \addplot+[
        blue,
        mark=none,
        thick
    ] table[
        col sep=comma,
        skip first n=1,
        x index=0,
        y index=10
    ]{data/semantic_support_taylor.csv};

    \addplot+[
        red,
        mark=none,
        dashed,
        thick
    ] table[
        col sep=comma,
        skip first n=1,
        x index=0,
        y index=20
    ]{data/semantic_support_taylor.csv};

    % ----------------------------------------------------------
    % Top-right: Instructions-2M
    % ----------------------------------------------------------
    \nextgroupplot[
        title={Instructions-2M},
        yticklabels={},
    ]

    \addplot+[
        color=blue!60,
        mark=none,
        thick
    ] table[
        col sep=comma,
        skip first n=1,
        x index=0,
        y index=2
    ]{data/semantic_support_instructions2m.csv};

    \addplot+[
        color=red!60,
        mark=none,
        dashed,
        thick
    ] table[
        col sep=comma,
        skip first n=1,
        x index=0,
        y index=12
    ]{data/semantic_support_instructions2m.csv};

    \addplot+[
        blue,
        mark=none,
        thick
    ] table[
        col sep=comma,
        skip first n=1,
        x index=0,
        y index=10
    ]{data/semantic_support_instructions2m.csv};

    \addplot+[
        red,
        mark=none,
        dashed,
        thick
    ] table[
        col sep=comma,
        skip first n=1,
        x index=0,
        y index=20
    ]{data/semantic_support_instructions2m.csv};

    % ----------------------------------------------------------
    % Bottom-left: LMSYS Chat
    % Includes title, y-axis label, and y-axis tick labels
    % ----------------------------------------------------------
    \nextgroupplot[
        title={LMSYS Chat},
        ylabel={Prob @$\psi r$},
    ]

    \addplot+[
        color=blue!60,
        mark=none,
        thick
    ] table[
        col sep=comma,
        skip first n=1,
        x index=0,
        y index=2
    ]{data/semantic_support_lmsys_chat.csv};

    \addplot+[
        color=red!60,
        mark=none,
        dashed,
        thick
    ] table[
        col sep=comma,
        skip first n=1,
        x index=0,
        y index=12
    ]{data/semantic_support_lmsys_chat.csv};

    \addplot+[
        blue,
        mark=none,
        thick
    ] table[
        col sep=comma,
        skip first n=1,
        x index=0,
        y index=10
    ]{data/semantic_support_lmsys_chat.csv};

    \addplot+[
        red,
        mark=none,
        dashed,
        thick
    ] table[
        col sep=comma,
        skip first n=1,
        x index=0,
        y index=20
    ]{data/semantic_support_lmsys_chat.csv};

    % ----------------------------------------------------------
    % Bottom-right: Yelp
    % ----------------------------------------------------------
    \nextgroupplot[
        title={Yelp},
        yticklabels={},
        legend style={
            font=\tiny,
            at={(0.03,0.97)},
            anchor=north west,
            draw=black,
            fill=white,
            legend cell align=left,
        },
    ]

    \addplot+[
        color=blue!60,
        mark=none,
        thick
    ] table[
        col sep=comma,
        skip first n=1,
        x index=0,
        y index=2
    ]{data/semantic_support_yelp.csv};

    \addplot+[
        color=red!60,
        mark=none,
        dashed,
        thick
    ] table[
        col sep=comma,
        skip first n=1,
        x index=0,
        y index=12
    ]{data/semantic_support_yelp.csv};

    \addplot+[
        blue,
        mark=none,
        thick
    ] table[
        col sep=comma,
        skip first n=1,
        x index=0,
        y index=10
    ]{data/semantic_support_yelp.csv};

    \addplot+[
        red,
        mark=none,
        dashed,
        thick
    ] table[
        col sep=comma,
        skip first n=1,
        x index=0,
        y index=20
    ]{data/semantic_support_yelp.csv};

    \legend{Fre-E2T, Aug-PE}

    \end{groupplot}
    \end{tikzpicture}
\caption{\textbf{Empirical semantic support protection.}
    Probability (lower is better) that an infrequent text has a neighbor
    among the synthetic texts within distance $\psi r$ where $\psi$ sweeps in $[0.1, 1.5]$. Lighter and
    darker colors show results on $(r,t)$-infrequent texts with $t=10$
    and $t=100$, respectively.}
    \label{fig:semantic}
    \Description{Plots across four datasets show the probability that infrequent input texts have nearby synthetic neighbors as the distance threshold increases. Curves distinguish privacy budgets and infrequency thresholds of 10 and 100.}
\end{figure}

\begin{table}[t]
\setlength{\tabcolsep}{3.5pt}
\caption{Downstream classification accuracy on Yelp.}
\label{tab:downstream}
\begin{tabular}{lcccc}
\toprule
 & \multicolumn{4}{c}{$\varepsilon$} \\
\cmidrule(lr){2-5}
Method & $4$ & $8$ & $16$ & $\infty$ \\
\midrule
\name: nearest centroid, $k=1$ & 0.35 & 0.45 & 0.46 & 0.47 \\
\name: nearest centroid, best $k$ & 0.35 & 0.46 & 0.48 & 0.49 \\
\name: RoBERTa (paraphrase $T=0.7$) & 0.32 & 0.44& 0.46 & 0.48\\
\name: RoBERTa (paraphrase $T=1.0$) & 0.32 & 0.42& 0.46 & 0.47\\
\name: RoBERTa (paraphrase $T=1.4$) & 0.30 & 0.45& 0.47 & 0.51\\
\midrule
Aug-PE: RoBERTa (default $T=1.4$) & 0.49 & 0.51\, & 0.52\, & 0.52 \\
\midrule
Non-DP: RoBERTa & \multicolumn{4}{c}{0.63} \\
\bottomrule
\end{tabular}
\end{table}

\vspace{0pt}
\noindent\textbf{Semantic support protection.} Recall Definition~\ref{def:semantic-protect},
the parameters $r^*$ and $\delta^*$ control the strength of semantic support protection: larger $r^*$ and smaller $\delta^*$ impose stronger restrictions on the probability of outputs appearing near---i.e., within distance $r^*$---semantically infrequent texts, which are defined by $r$ and $t$ (recall Definition~\ref{def:euclidean-secret}). We empirically evaluate the level of $\delta^*$ from $(r,t,r^*,\delta^*)$-Semantic Support Protection in Figure~\ref{fig:semantic}. We fix $\varepsilon=4$ and vary $r^*=\psi r$ for $\psi\in [0.1 , 1.5]$ and $t$ from $10$ to $100$. We observe highly similar results on different $t$'s and report only for $t=10$ and $100$. In TaylorAI and Yelp, we achieve perfect protections of infrequent texts (i.e., $\delta^*=0$) that vanish in the log-scale plot. Also, fixing $r^*=r$ (namely, $\psi=1$), \name achieves a $\delta^*$ near $10^{-5}$ under the simplified Definition~\ref{def:semantic-protect-simplified} across all datasets. In general, the semantic support protection of our \name is much stronger than Aug-PE~\cite{xie2024differentially}, and sometimes \textit{by orders of magnitude}. \looseness=-1

In practice, we advice to set $r^*=r$ following Definition~\ref{def:semantic-protect-simplified}, with $\delta^*$ set to $10^{-5}$---a small failure probability that is also widely used in DP~\cite{abadi,DiscreteGaussian}. Regarding $r$, any value smaller than the reference distance from unrelated datasets makes sense. Here, we have chosen $r=0.5$ for an illustration. For $t$, we recommend a negligible number compared with the population. Here, each dataset contains around $1$ million records and we have chosen $t=10$ and $100$. 

\subsection{More Results on Text Quality}\label{sec:text-quality}

\vspace{0pt}
\noindent\textbf{Downstream utility.} For dataset Yelp with labels, we train a classifier to rate the released
texts, and evaluate
on the official $50{,}000$ test split, which is never used during
synthesis. The task is five-way star prediction. We consider two ways of using the output of \name. The first is directly operated on the released noisy centroids: a test review is embedded using the same model and classified by a majority vote over the $k$ nearest centroids released by Fre-E2T, across all $5$ labels. We sweep $k\in\{1,5,10,20\}$ and report the best result. The second uses the inverted texts from Fre-E2T to fine-tune the RoBERTa-base classifier, following the setup in~\cite{xie2024differentially}. For Fre-E2T, we paraphrase each text $5$ times, obtaining around $3000$ texts in total for each temperature. We also make Aug-PE output the same number of texts, to make it a fair comparison. \looseness=-1

Table~\ref{tab:downstream} reports the classification accuracy. The non-DP baseline fine-tuned on RoBERTa achieves accuracy of $0.63$ while the random guessing baseline is $0.20$. With $\varepsilon\ge 8$, our method achieves competitive performance as the competitor Aug-PE. Notably, the difference between using the inverted texts and the centroids directly is not large, which confirms the utility of the DP centroids obtained by \name. Still, there remains some gap to fill, and we conjecture that a better vector-to-text inverter could bridge this gap. We defer further discussions on the limitations and future work to Section~\ref{sec:design-choices}. \looseness=-1

\vspace{0pt}
\noindent\textbf{Diversity.} Each cluster of \name emits one text from the inversion model. However, diversity can be improved by paraphrasing and tuning the temperature of the paraphraser model (larger $T$ means more randomness and therefore higher diversity). In Table~\ref{tab:downstream}, we can see that the classification utility is largely unaffected by the paraphrasing temperature, which is aligned with~\cite{xie2024differentially}. Table~\ref{tab:bleu} shows that higher temperature leads to better diversity, measured as the Self-BLEU (smaller is more diverse) per cluster, averaged over all clusters. Still, paraphrasing does not fully close the gap to real texts, since each cluster ultimately emits only one centroid based on the users' contributions. We acknowledge this as a limitation.

\begin{table}[t]
\centering
\caption{Diversity in BLEU score ($\downarrow$) of Fre-E2T under different temperatures for paraphrasing, evaluated on Yelp.}
\label{tab:bleu}

\begin{tabular}{lccccc}
\toprule
\cmidrule(lr){2-5}
Temperature 
& $\varepsilon=4$
& $\varepsilon=8$
& $\varepsilon=16$
& $\varepsilon=\infty$
& Real \\
\midrule

$T=0.7$
& 52.3
& 50.8
& 48.4
& 46.4
&   \\

$T=1.0$
& 49.6
& 47.3
& 46.3
& 45.4
& 27.7 \\

$T=1.4$
& 47.5
& 44.1
& 42.1
& 41.7
& \\

\bottomrule
\end{tabular}
\end{table}

\vspace{0pt}
\noindent\textbf{Perpelxity.} We also evaluate the released texts directly against matched-size random samples from
the frequent subset of the original input, based on the perplexity score (smaller is more natural) computed from GPT-2. As we can see from Table~\ref{tab:ppl-main}, the average perplexity score of Fre-E2T improves monotonically as $\varepsilon$ grows and
approaches the real text samples. Note that the average metric can be affected by outliers. We show that most of the texts are deemed natural---namely, they have a low perplexity---in Appendix~\ref{app:diversity} (see also the formal definition of perplexity).

\begin{table}[t]
\centering
\caption{Average perplexity ($\downarrow$) of \name evaluated with GPT-2. }
\label{tab:ppl-main}

\begin{tabular}{lccccc}
\toprule
\cmidrule(lr){2-5}
Dataset
& $\varepsilon=4$
& $\varepsilon=8$
& $\varepsilon=16$
& $\varepsilon=\infty$
& Real \\
\midrule

TaylorAI
& 196
& 125
& 74
& 51
& 22 \\

Instructions-2M
& 95
& 42
& 26
& 26
& 57 \\

LMSYS Chat
& 135
& 87
& 52
& 31
& 30 \\

Yelp
& 184
& 90
& 51
& 34
& 27 \\

\bottomrule
\end{tabular}
\end{table}

\newcommand{\ps}{p_{s}}
\newcommand{\epsfre}{\varepsilon_{\mathrm{fre}}}
\newcommand{\dfre}{\delta_{\mathrm{fre}}}
\newcommand{\dsens}{\delta_{\mathrm{sens}}}
\newcommand{\dtot}{\delta_{\mathrm{tot}}}
\newcommand{\Inst}{Instructions-2M}
\newcommand{\Lmsys}{LMSYS Chat}
\newcommand{\best}[1]{\textbf{\boldmath #1}}

\subsection{Ablation Studies}\label{app:ablation}
We ablate each parameter of Fre-E2T on \Inst{} and \Lmsys{}. Unless otherwise stated, parameters take their default values $k=20$, $t=100$, bucket edge
$L=2r/\sqrt{k}$, $r=0.5$ and we conduct experiments mostly on $\varepsilon=8$ to observe a stable trend. Results are evaluated on text level, without GPT-2 paraphrasing, unless otherwise stated. Finally, we recall the couplings of the parameters: the heavy-hitter threshold enters the accounting only through
$\tau=\ps\,t$, and the sensitivity bound $\Delta=u\,r$ is paid for with
$\dsens=\big(f(2/u)\big)^{k}$, so $\Delta$ must be calibrated together with $L$ and $k$.

\vspace{0pt}\noindent\textbf{Subsampling rate $\ps$.} At fixed $t$, $\ps$ moves both terms of the budget: $\epsfre=v\ln(1/(1-\ps))$
grows with $\ps$, while $\dfre$ shrinks with $\tau=\ps t$. Here, $v$ controls the budget split between the heavy hitter histogram and the centroid computation. The feasible region at
overall $\delta=10^{-6}$ is therefore bounded on {both} sides: $\ps\le0.1$ exhausts
the $\delta$ budget for every $v\le5$ ($\dfre\ge4\cdot10^{-5}$). Inside the feasible region, the optimal $\ps$ is near $0.5$, as Table~\ref{tab:ps} shows. For each $\ps$ we have chosen the budget split $v$ that minimizes the
calibrated $\sigma$. The $v=\{5,5,4,4,3\}$ used for the main experiments is not that choice: at
$\eps=8$, $v=2$ raises F1 from
$0.36$ to $0.42$ on \Inst{} and from $0.64$ to $0.71$ on \Lmsys{}. We advise to choose {$\ps$} as large as the $\eps$ budget allows and confirm that we did not specifically fine-tune the parameters of \name to gain advantage over Aug-PE.

\begin{table}[t]
\centering
\caption{Subsampling rate $\ps$ at $t=100$ and $\eps=8$. Parameter $v$ is chosen to minimize the noise scale $\sigma$.}
\label{tab:ps}

\begin{tabular}{cccc}
\toprule
$\ps$ & $(v,\sigma)$ & F1 \Inst{} & F1 \Lmsys{} \\
\midrule

$0.15$ & $(5,\,0.914)$ & $0.35$          & $0.45$ \\
$0.20$ & $(4,\,0.873)$ & $0.33$          & $0.51$ \\
$0.30$ & $(3,\,0.886)$ & $0.37$          & $0.61$ \\
$0.50$ & $(2,\,0.923)$ & $0.42$   & $0.71$ \\
$0.70$ & $(2,\,1.067)$ & $0.35$          & $0.74$ \\

\bottomrule
\end{tabular}
\end{table}

\vspace{0pt}\noindent\textbf{Bucket edge $L$.} We scale $L=2r/\sqrt{k}$ by $m\in\{0.5,\dots,2\}$ with $\Delta=2.4\,rm$, which
keeps the failure probability of sensitivity bound $\dsens=8.9\cdot10^{-11}$ invariant. Accordingly, we re-calibrate $\sigma$, which is exactly proportional to $m$. We list the values of $\sigma$ in Table~\ref{tab:L-sigma}. The results are shown in Table~\ref{tab:L}. With differential privacy ($\varepsilon<\infty$), smaller buckets win consistently: halving $L$ halves the
noise, and still there are enough buckets that reach the threshold $\tau$. At $\eps=\infty$, F1 score first increases and then decreases as the $L$ increases. This is because without noise, coarser buckets improves the recall while reducing the precision in general. The default $m=1$ in our main setup is not fine-tuned to optimize our Fre-E2T, ensuring a fair comparison with Aug-PE. 

\begin{table}[t]
\caption{F1 score under different bucket-edge $L$ scaled by multiplier $m$.}
\label{tab:L}
\begin{tabular}{llccccc}
\toprule
& & \multicolumn{5}{c}{$m$}\\
\cmidrule(l){3-7}
$\eps$ & Dataset & $0.5$ & $0.75$ & $1$ & $1.5$ & $2$\\
\midrule
$2$ & \Inst{}      & $0.18$ & {$0.18$} & $0.16$ & $0.14$ & $0.12$\\
    & \Lmsys{}     & {$0.34$} & $0.29$ & $0.24$ & $0.24$ & $0.25$\\
$8$ & \Inst{}      & {$0.41$} & $0.37$ & $0.35$ & $0.29$ & $0.28$\\
    & \Lmsys{}     & {$0.71$} & $0.69$ & $0.66$ & $0.64$ & $0.60$\\
$\infty$ & \Inst{} & $0.50$ & $0.50$ & $0.50$ & {$0.76$} & $0.72$\\
    & \Lmsys{}     & $0.78$ & $0.79$ & $0.83$ & {$0.86$} & $0.86$\\
\bottomrule
\end{tabular}
\end{table}

\vspace{0pt}\noindent\textbf{Embedding model.} We repeat the pipeline with another open-sourced embedding model \texttt{gte-modernbert-base}~\cite{gte}. Since embedding norm scales are
not comparable across different embeddings---gte has median norms $\approx38$ while gtr has norms around $0.4$--$2.6$, so we transfer the radius by quantile matching. We choose $r'$ that matches the quantile of nearest-neighbour distance distribution that $r=0.5$ produces under the gtr embedding. This
gives $r'=14.86$ on \Inst{} and $r'=15.84$ on \Lmsys{}. Since no public vector-to-text inverter exists for this embedding family, the results are evaluated at the centroid level, shown in Table~\ref{tab:emb}.

\begin{table}[t]
\caption{Utility with different embedding models, measured at the centroid level.}
\label{tab:emb}
\centering
\begin{tabular}{@{}llcc@{\hspace{0.75em}}cc@{}}
\toprule
& & \multicolumn{2}{c}{gtr-t5-base} & \multicolumn{2}{c}{gte-modernbert-base}\\
\cmidrule(lr){3-4}\cmidrule(l){5-6}
Dataset & $\eps$ & Precision & Recall & Precision & Recall \\
\midrule
\Inst{}
 & $2$   & $0.29$ & $0.39$ &  $0.30$ & $0.65$ \\
 & $8$        & $1.00$ & $0.47$ &  $1.00$ & $0.81$ \\
 & $\infty$   & $1.00$ & $0.48$ &  $1.00$ & $0.82$ \\
\midrule
\Lmsys{}
 & $2$        & $0.03$ & $0.82$ &  $0.09$ & $0.40$ \\
 & $8$        & $0.99$ & $0.95$ & $1.00$ & $0.66$ \\
 & $\infty$   & $0.99$ & $0.95$ & $0.99$ & $0.67$\\
\bottomrule
\end{tabular}
\end{table}

\section{Design Choices and Limitations}\label{sec:design-choices} 

% \subsection{Justification of Our Design}\label{sec:justify}

Fre-E2T essentially performs two 
tasks: identifying heavy hitting regions in the embedding space and 
inverting from the heavy hitters to synthetic texts. We first analyze the utility on the first task---identifying heavy hitting regions
in the embedding space with DP. The first task is supported
by two major technical components: random projection and partitioning  and DP heavy hitter estimation,
followed by the DP centroid computation. 

\vspace{0pt}
\noindent\textbf{Comparison with LDP mechanisms.} For a clear reference, we evaluate two locally differentially private
(LDP) baselines and a centralized-DP baseline. The first LDP baseline, named \textit{LDP (embedding)}, is to let each user release its text embedding vector with
Gaussian noise injected~\cite{balle2018improving}. Any post-processing on the
noisy embeddings does not incur additional privacy cost~\cite{dpbook}. The second LDP baseline, named \textit{LDP (OUE)}, is to replace the sample-and-threshold
subroutine (lines 7--15 in Algorithm~\ref{alg:centralized}) with the OUE frequency
estimation algorithm~\cite{oue}. The rest follows: surviving buckets compute the
noisy centroid with additive Gaussian noise with the remaining privacy budget. Here, we report
the privacy budget split that maximizes utility.\looseness=-1

The centralized DP baseline removes every distributed constraint while keeping
the rest of the pipeline unchanged: a trusted curator observes all embeddings, releases
each occupied bucket's count with Gaussian noise 
% (sensitivity $1$) 
and its embedding sum with Gaussian noise under the same
probabilistic sensitivity bound $\Delta_{dist}$ as ours, then keeps the buckets whose noisy count exceeds $\tau=100$. This baseline upper-bounds what any distributed
realization of our pipeline can achieve. \looseness=-1

Table~\ref{tab:ldp-identification} reports how well each mechanism identifies
the heavy clusters in terms of precision and recall. The first LDP baseline
identifies nothing at any $\varepsilon\le 16$: its local noise is $O(\sqrt{d})$, which is the same scale as the diameter of the whole embedding space, so no two users ever share a bucket and no
cluster survives the threshold. The second fails for the
same structural reason: the OUE count estimator's noise is in $O(\sqrt{n})$, where $n$ is the number of users, which makes the heavy hitter histogram very noisy. As a result, the recall is only 
$3\%$ and $14\%$ for TaylorAI and Instructions-2M, respectively at $\varepsilon=4$. It reaches to about half at $\varepsilon=8$, and catches
up with Fre-E2T only at $\varepsilon=16$. Fre-E2T, in
contrast, identifies the heavy clusters with precision and recall above $0.90$ at \emph{every} $\varepsilon$. 
% Results on other $\tau$'s are similar and omitted.

These results justify our design. First,
\emph{sample-and-threshold} solves the identification task at almost no cost
relative to a trusted curator. The LDP alternatives show the utility loss. Second,
\emph{probabilistic sensitivity} is what makes the centroid sums achieve high utility at
all $\varepsilon\in\{4,8,16\}$. A local randomizer must be indistinguishable over the entire input space, introducing large noises that are in the same scale of the space's diameter, overwhelming useful information. Our release instead adds noise proportional to the probabilistic sensitivity. \looseness=-1

\begin{table}[t!]
\centering
\setlength{\tabcolsep}{5pt}
\caption{Precision and recall measured at cluster level for heavy hitting clusters with $\tau\ge 100$.}
\label{tab:ldp-identification}
\begin{tabular}{@{}llrrrr@{}}
\toprule
& & \multicolumn{2}{c}{\textbf{TaylorAI}} & \multicolumn{2}{c@{}}{\textbf{Instructions-2M}}\\
\cmidrule(lr){3-4}\cmidrule(l){5-6}
Method & $\varepsilon$ & Prec. & Rec. & Prec. & Rec. \\
\midrule
LDP (Embedding)  & 4  & 0.00  & 0.00 & 0.00  & 0.00 \\
LDP (OUE)   & 4  & 0.03 & 0.03 & 0.14 & 0.14 \\
Fre-E2T  & 4  & 0.93 & 0.91 & 0.92 & 0.92 \\
Central-DP  & 4  & 0.97 & 0.96 & 1.00 & 1.00 \\
\cmidrule(l){1-6}
LDP (Embedding)  & 8  & 0.00  & 0.00 & 0.00  & 0.00 \\
LDP (OUE)   & 8  & 0.51 & 0.52 & 0.56 & 0.57 \\
Fre-E2T  & 8  & 0.94 & 0.95 & 0.95 & 0.98 \\
Central-DP  & 8  & 0.98 & 0.98 & 1.00 & 1.00 \\
\cmidrule(l){1-6}
LDP (Embedding)  & 16 & 0.00  & 0.00 & 0.00  & 0.00 \\
LDP (OUE)   & 16 & 0.98 & 0.95 & 0.98 & 0.98 \\
Fre-E2T  & 16 & 0.98 & 0.96 & 0.99 & 0.99 \\
Central-DP  & 16 & 1.00 & 0.99 & 1.00 & 1.00 \\
\bottomrule
\end{tabular}
\end{table}

\begin{table}[t!]
\centering
\setlength{\tabcolsep}{5pt}
\caption{Precision and recall of the synthetic texts.}
\label{tab:inversion-bottleneck}
\begin{tabular}{@{}llrrrr@{}}
\toprule
& & \multicolumn{2}{c}{\textbf{TaylorAI}} & \multicolumn{2}{c@{}}{\textbf{Instructions-2M}}\\
\cmidrule(lr){3-4}\cmidrule(l){5-6}
Method & $\varepsilon$ & Prec. & Rec. & Prec. & Rec. \\
\midrule
Fre-E2T  & 4  & 0.31 & 0.92 & 0.25 & 0.32 \\
Central-DP  & 4  & 0.72 & 0.97 & 0.48 & 0.37 \\
\cmidrule(l){1-6}
Fre-E2T  & 8  & 0.62 & 0.95 & 0.38 & 0.34 \\
Central-DP  & 8  & 0.83 & 0.97 & 0.62 & 0.38 \\
\cmidrule(l){1-6}
Fre-E2T  & 16 & 0.82 & 0.97 & 0.56 & 0.37 \\
Central-DP  & 16 & 0.86 & 0.97 & 0.62 & 0.38 \\
\bottomrule
\end{tabular}
\end{table}

\vspace{0pt}
\noindent\textbf{Representativeness.} Representativeness concerns which population the output speaks for. The released
texts hold only $41\%$ and $3.9\%$ of the \emph{entire} corpus, but $97\%$ and $39\%$ of
the frequent subset of TaylorAI and $4.5\%$ Instructions-2M, respectively, matching the numbers in Table~\ref{tab:combined_dp_results}. This is by our design, we filter out infrequent buckets to achieve SSP while achieving high recall and precision for the frequent user texts.

\vspace{0pt}
\noindent\textbf{Limitation and future work.} 
We also identify a key limitation of our Fre-E2T. Table~\ref{tab:inversion-bottleneck} scores the synthetic texts obtained after inversion for Fre-E2T. Comparing Tables~\ref{tab:inversion-bottleneck} and~\ref{tab:ldp-identification}, we can see that the embedding-to-text inversion, rather than the DP
machinery, is the current bottleneck of Fre-E2T, as we observe a utility drop on the inverted text. In particular, for Instructions-2M dataset, when $\varepsilon=8$. the precision and recall dropped from $0.95$ and $0.98$ to $0.38$ and $0.34$. For the centralized DP method that uses the same inversion model, this drop is still substantial. At $\varepsilon=16$, Fre-E2T achieves almost the same utility as its centralized-DP version. This confirms that the utility drop is from the inverter~\cite{morris-etal-2023-text} that we used as a black box, not the distributed constraints. It was also noted in~\cite{morris2023language} that even inverting from noiseless embeddings may not recover the perfect text. Overall, domain-adapted inverters and designs that capture the objectives of diversity and downstream task utility (recall Section~\ref{sec:text-quality}) are
therefore the promising directions for improving the end-to-end quality of Fre-E2T.

\balance
\section{Conclusion}
\label{sec:conclusion}
We propose a lightweight distributed algorithm for synthesizing texts from
distributed users. Our
algorithm design does not rely on trusted servers, gradient/parameter access
to LMs, iterative and tight user synchronization, or publicly available
labels, addressing major limitations of the existing literature. Our method can also be implemented efficiently based on existing distributed protocols. Compared
with the state of the art, our method achieves comparable and sometimes
better utility across real-world datasets, with an extra layer of
privacy protection for infrequent user data and better practicality in
distributed settings. 

\looseness=-1

\clearpage
\nobalance

%==============================================================================
% Everything below the main body — including appendices, the three mandatory
% PoPETs sections, and the bibliography — does NOT count toward the page
% limit.
%==============================================================================
\appendix

% ----- Mandatory PoPETs disclosure sections (placed between \appendix and the
%        bibliography; produced by the popets.sty environments) -------------
%==============================================================================
% Mandatory PoPETs 2027 sections.
% Do not count toward the 13-page camera-ready main-body limit. Use the popets.sty
% environments (ethics, openscience, ai) and be placed between \appendix /
% acks and the bibliography, per https://petsymposium.org/authors-2027.php.
%==============================================================================
\begin{acks}
This work is supported by the Ministry of Education, Singapore, under Tier-2 Grant MOE-000761-01 and Tier-1 Grant T1 251RES2307, and a Google Gift Grant.

We thank the reviewers and editors for their constructive feedback during the review process.
% TODO(author): Add the required funding acknowledgment covering all authors.
% If there was no specific funding, use the exact no-funding statement in the
% PoPETs camera-ready instructions after all authors have confirmed this.
\end{acks}

\begin{ethics}
We have read and adhered to the PoPETs ethics guidelines and the Menlo Report
and believe the work was conducted ethically. The work does not involve human
subjects, sensitive user data, or system vulnerability disclosure; all
experiments use publicly available benchmark datasets and openly released
language models.

\paragraph{Stakeholders.}
The primary stakeholders are:
\begin{itemize}
\item \textbf{End users.} Our protocol is designed so that distributed users
  can collectively contribute to a synthetic text corpus without any party,
  including the server that runs the protocol, observing their raw texts.
  Users benefit from formal $(\varepsilon,\delta)$-DP guarantees that protect
  individual records and additional protection for users whose texts are
  semantically infrequent.
\item \textbf{Researchers.} Our work enables the research community to
  collaboratively generate synthetic texts from distributed private corpora,
  facilitating research collaboration while preserving the privacy of the
  underlying data.
\item \textbf{System developers.} The algorithm can be deployed in
  privacy-preserving data collection and analysis pipelines, enabling service
  providers to improve models and user experience without ingesting raw
  user text.
\end{itemize}

\paragraph{Ethical principles.}
Following the Menlo Report, we adhere to four principles:
\begin{itemize}
    \item \textbf{Beneficence.} The purpose of this research is to improve the
    practicality and the theory of synthesizing texts from distributed users
    without degrading privacy or utility. Experiments are restricted to public
    benchmarks.
    \item \textbf{Respect for persons.} No personal or private user data was
    used. All datasets and models were either synthetic or publicly available.
    \item \textbf{Justice.} We highlight the scalability issue of existing DP
    text-synthesis methods and provide a more practical distributed solution.
    We hope our work can be used by system developers to facilitate
    privacy-preserving data collection and analysis.
    \item \textbf{Respect for law and public interest.} All experiments comply
    with the applicable open-source software licenses.
\end{itemize}

\paragraph{Potential harms and mitigations.}
\begin{itemize}
    \item \textbf{Risk of over-interpretation.} Our contribution is an
    algorithm for DP text synthesis; it is not intended to be used to generate
    fake content for misuse.
    \item \textbf{Controlled scope.} Our work is restricted to the differential
    privacy framework under the semi-honest threat model. The algorithm should
    not be deployed in arbitrary environments outside this assumption set, as
    that may lead to data exploits not covered by the analysis.
\end{itemize}

\paragraph{Decision to conduct and publish.}
We believe that proposing a
practical algorithm for distributed differentially private text synthesis
is valuable to the community, since our algorithm enables researchers and
system developers to share, collect, and analyse sensitive textual data
without exposing individual records or violating user privacy.

\paragraph{External review.}
This work was not submitted to an external ethics panel (e.g., IRB or the Tor
Research Safety Board) because the study uses only public benchmarks and does
not involve human subjects or sensitive user data.
\end{ethics}

\begin{openscience}
All datasets used in this work are
publicly available. Upon acceptance, the artifacts will be released, including: (1) 
implementations of the 
algorithm; (2) all dataset preprocessing scripts; and (3) any scripts essential to reproduce the results. 
\end{openscience}

\begin{ai}
We disclose all use of
generative AI in the preparation of this manuscript and the underlying
research. The authors used ChatGPT to revise the text to improve flow, correct typographical and grammatical issues, and improve the presentation of figures, based on experimental results. We verified that AI did not introduce new claims, and did not fabricate references. Generative AI was \emph{not} used to design the algorithm or derive the privacy
analysis. All references were collected from the original publications.
\end{ai}

% ----- Bibliography --------------------------------------------------------
\bibliographystyle{ACM-Reference-Format}
\bibliography{main}

\clearpage

\section{Additional Experiment Details and Results}\label{app:exp}

We have used the following datasets.

\vspace{0pt}
\noindent \textbf{1. TaylorAI}~\cite{tayloraiuserqueriesdataset} includes $1.5$ million user queries.

\vspace{0pt}
\noindent\textbf{2. Instructions-2M}~\citep{morris2023language} is a meta-dataset containing approximately $2.33$ million instructions, including user and system prompts spanning different problem domains.

\vspace{0pt}
\noindent \textbf{3. LMSYS Chat-1M}~\citep{zheng2023lmsyschat1m} contains one million real-world conversations with 25 state-of-the-art LLMs. The dataset includes interactions collected from 210K unique IP addresses. We consider each user query as an individual.

\vspace{0pt}
\noindent\textbf{4. Yelp}~\cite{zhang2015character} consists of reviews from Yelp. It is extracted from the Yelp Dataset Challenge 2015 data.

We apply a preprocessing step to filter out samples where the language is non-English, minimizing the potential impact of multilingual limitations in the embedding models. The resulting statistics are summarized in Table~\ref{tab:dataset_stats}.

\begin{table}[h]
\centering
\caption{Dataset statistics for the overall number of samples and number of frequent items \textbf{in millions} and the average token length per sample.}
\label{tab:dataset_stats}
\begin{tabular}{lcccc}
\toprule
\textbf{Dataset} & \textbf{\#samples} & \textbf{\#texts} & \textbf{Average length} \\
\midrule
TaylorAI & 1.29$\times 10^{6}$ & 0.442$\times 10^{6}$ & 88.1   \\
Instructions-2M & 2.33$\times 10^{6}$  & 0.229$\times 10^{6}$ & 30.6 \\
LMSYS-Chat-1M & 1.50$\times 10^{6}$ & 0.446$\times 10^{6}$ & 72.0   \\
Yelp & 0.65$\times 10^{6}$  & 0.518$\times 10^{6}$ & 192 \\
\bottomrule
\end{tabular}
\end{table}

\subsection{Hyperparameters}\label{app:parameter-detail}
\begin{algorithm}[t!]
\caption{\textsc{CalibrateNoise}}
\DontPrintSemicolon
\label{alg:calibratenoise}
\KwIn{target $(\varepsilon, \delta)$; radius $r$; frequency threshold $t$;
projection dimension $k$; subsampling rate $p_s$; budget factor $v \ge 1$;
sensitivity ratio $u$ (so $\Delta_{\mathrm{dist}} = u\,r$)}
\KwOut{Gaussian noise scale $\sigma$ for the centroid sums}
$\tau \gets p_s\, t$ \tcp*{threshold on sampled counts}
$\varepsilon_{fre} \gets v \ln\frac{1}{1-p_s}$;\;
$q \gets 1 - (1-p_s)^{v+1}$\;
$\delta_{fre} \gets \exp\!\big(-\tfrac{\tau}{q}\, D(q \,\|\, p_s)\big)$
  \tcp*{Lemma~\ref{lem:s-and-t-dp}}
$\delta_{sens} \gets f(2/u)^{k}$ \tcp*{Lemma~\ref{lem:high-prob-sens} }
$\varepsilon_{agg} \gets \varepsilon - \varepsilon_{fre}$\;
$\delta_{agg} \gets \delta - \delta_{fre} - \delta_{sens}$\tcp*{budget for Gaussian}
\If{$\varepsilon_{agg} \le 0$ \textbf{\emph{or}} $\delta_{agg} \le 0$}{
  abort \tcp*{lower $v$ (or $p_s$) if $\varepsilon_{agg} \le 0$;
              raise $\tau$, $v$, or $u$ if $\delta_{agg} \le 0$}}
$\rho_{agg} \gets \max\Big\{\, \rho > 0 :
  \inf_{\alpha \in (1,\infty)}
    \frac{\exp\!\big((\alpha-1)(\alpha \rho - \varepsilon_{agg})\big)}{\alpha - 1}
    \Big(1 - \frac{1}{\alpha}\Big)^{\alpha}
  \;\le\; \delta_{agg} \,\Big\}$
  \tcp*{Thm.~\ref{thm:main-gaussian}}
\Return{$\sigma \gets u\,r / \sqrt{2 \rho_{agg}}$}\;
\end{algorithm}

Our algorithm rely on the following hyperparameters. 
\begin{itemize}[leftmargin=*]
\setlength{\itemsep}{0ex}
\item \textbf{Overall privacy parameters.} We fix $\delta=10^{-6}$ and vary $\varepsilon\in\{4,8,16\}$ and $\infty$ (i.e., non-DP setting).
\item Recall that for any two data points in the embedding space $\mathbb{R}^d$, we see them as $r$-neighbors if their $\mathcal{L}_2$ distance is not larger than $r$. We say a text is frequent if it has at least $t$ neighbors in its $r$-neighborhood. We have taken $r=0.5$ and $t=100$ unless otherwise stated in our experiments (see Section~\ref{sec:exp}).

\item \textbf{Data pre-processing for embedding, projection, and partitioning.}
Each client calls the public, black-box embedding function to compute a $d$-dimensional embedding vector for
its text, with $d = 768$. Next, the $d$-dimensional vector is randomly projected
onto the $k$-dimensional space using a Gaussian projector
$\in \mathbb{R}^{d \times k}$ whose entries are sampled independently from
$\mathcal{N}(0, 1)$, scaled by $1/\sqrt{k}$. We fix $k = 20$ in our experiments
and provide ablation studies on the impact of $k$ in Section~\ref{sec:exp}. The space
$\mathbb{R}^k$ is then partitioned into buckets of edge length
$L = 2r/\sqrt{k}$ on each dimension, with an independent uniform random offset
per dimension.

\item \textbf{Component for identifying heavy hitters.}
Each client contributes its bucket index to the synthesis server with
probability $p_s$ (the subsampling rate). The server keeps the bucket indices
whose sampled count $n_h$ reaches the threshold $\tau = p_s \cdot t$, i.e.\
$\tau = 30, 50, 60, 100$ for
$\varepsilon =  4, 8, 16, \infty$ with
$p_s = 0.3, 0.5, 0.6, 1$. Following Lemma~\ref{lem:s-and-t-dp}, we spend
$\varepsilon_{fre} = v \cdot \log\frac{1}{1-p_s}$ on this release, with the
budget factor $v = 4, 4, 3$ for $\varepsilon = 4, 8, 16$; the
corresponding
$\delta_{fre} = \exp\!\big(-\tfrac{\tau}{q}\, D(q \,\|\, p_s)\big)$,
$q = 1 - (1-p_s)^{v+1}$, is listed in Table~\ref{tab:delta-split}. Here $v$ was chosen per target so that the remaining $\delta$ budget
suffices, see Algorithm~\ref{alg:calibratenoise}.

\item \textbf{Component for centroid computation.}
For each heavy bucket, the clients that submitted to it aggregate their
$d$-dimensional embedding vectors, and Gaussian noise
$\mathcal{N}(0, \sigma^2 \mathbf{I}_d)$ is added to the sum. We set the target
sensitivity to $\Delta_{\mathrm{dist}} = u \cdot r$; the ratio $u$ is fixed per
projection dimension. We set $u = 2.4$ at our default $k = 20$, and its
failure probability $\delta_{sens} = f(2/u)^k$ is accounted in the final
$\delta$ (Theorem~\ref{thm:dp-final-fre-e2t}). Given the split
$(\varepsilon_{fre}, \delta_{fre}, \delta_{sens})$, the noise scale is
calibrated with the Gaussian mechanism
(Theorem~\ref{thm:main-gaussian}) via procedure \textsc{CalibrateNoise}
(Algorithm~\ref{alg:calibratenoise}); for $k = 20$ this gives
$\sigma = 2.136, 1.14, 0.516$ for $\varepsilon = 4, 8, 16$
(noise multipliers $m = \sigma/\Delta_{\mathrm{dist}}$ of
$1.78, 0.95, 0.43$). When $\varepsilon = \infty$ (non-DP),
$\sigma = 0$ and no subsampling is applied
($p_s = 1$).
\end{itemize}

\vspace{0pt}
\noindent\textbf{Split of $\delta$.} Table~\ref{tab:delta-split} shows the split of the overall $\delta=10^{-6}$ (at $k = 20$) among the 
heavy-hitter release ($\delta_{fre}$, sensitivity-bound failure $\delta_{sens}$, and the remainder available to the centroid computation. We can see that $\delta_{sens}$ consumes a negligible portion of the overall budget---only around $0.1\%$. 

\begin{table}[h]
\centering
\caption{Split of the overall budget $\delta = 10^{-6}$ at $k = 20$.}
\label{tab:delta-split}
\begin{tabular}{cccc}
\toprule
$\varepsilon$ & $\delta_{fre}$ & $\delta_{sens}$ &
$\delta_{agg}$ budget \\
\midrule
4  & $2.93\times10^{-10}$ & $8.95\times10^{-11}$ & $1.00\times10^{-6}$ \\
8  & $3.80\times10^{-13}$ & $8.95\times10^{-11}$ & $1.00\times10^{-6}$ \\
16 & $1.76\times10^{-11}$ & $8.95\times10^{-11}$ & $1.00\times10^{-6}$ \\
\bottomrule
\end{tabular}
\end{table}

\vspace{0pt}
\noindent\textbf{Numerical procedure.} One can use Algorithm~\ref{alg:calibratenoise} to obtain the Gaussian noise scale given target parivacy parameters, radius $r>0$, frequency threshold $t\in \mathbb{N}^*$, projection dimension $k\in \mathbb{N}^*$, subsampling rate $p_s\in(0,1)$, the budget factor $v>0$, and the sensitivity ratio/noise multiplier $u>0$. Note that we use the continuous Gaussian mechanism (Theorem~\ref{thm:main-gaussian}) instead of the discrete distributed Gaussian mechanism (Theorem~\ref{thm:main-dgg-formal}) for simulation purpose, since~\cite{DiscreteGaussian} have pointed out that when the quantization granularity is high enough, both mechanisms are indistinguishable in utility under the same privacy budget.

\begin{table}[t!]
\centering\small
\caption{Noise scale $\sigma$ under different bucket-edge $L$ scaled by multiplier $m$.}
\label{tab:L-sigma}
\begin{tabular}{llccccc}
\toprule
& \multicolumn{5}{c}{$m$}\\
\cmidrule(l){2-6}
$\eps$  & $0.5$ & $0.75$ & $1$ & $1.5$ & $2$\\
\midrule
$2$         & $2.84$ & $4.26$ & $5.69$ & $8.53$ & $11.4$\\
$8$         & $0.566$ & $0.848$ & $1.13$ & $1.70$ & $2.26$\\
$\infty$        & 0 & 0 & 0 & 0 & 0\\
\bottomrule
\end{tabular}
\end{table}

\begin{table}[t!]
\caption{F1 score under varying projection dimension $k$.}
\label{tab:k}
\begin{tabular}{llccccc}
\toprule
& & \multicolumn{5}{c}{$k$}\\
\cmidrule(l){3-7}
dataset & $\eps$ & $10$ & $15$ & $20$ & $25$ & $30$\\
\midrule
\Inst{}
 & $2$           & $0.11$ & $0.15$ & $0.17$ & $0.17$ & {$0.18$}\\
 & $4$           & $0.13$ & $0.22$ & $0.28$ & {$0.31$} & $0.30$\\
 & $8$           & $0.26$ & $0.34$ & $0.36$ & {$0.44$} & $0.41$\\
 & $16$          & {$0.49$} & $0.47$ & $0.45$ & $0.40$ & $0.47$\\
 & $\infty$      & {$0.69$} & $0.54$ & $0.50$ & $0.50$ & $0.49$\\
\midrule
\Lmsys{}
 & $2$           & $0.25$ & $0.26$ & $0.25$ & {$0.29$} & $0.28$\\
 & $4$           & $0.38$ & $0.419$ & $0.43$ & $0.50$ & {$0.52$}\\
 & $8$           & $0.54$ & $0.637$ & $0.64$ & {$0.71$} & $0.70$\\
 & $16$          & $0.75$ & $0.77$ & {$0.80$} & $0.79$ & $0.76$\\
 & $\infty$      & {$0.85$} & $0.84$ & $0.83$ & $0.83$ & $0.79$\\
\bottomrule
\end{tabular}
\end{table}

\begin{table}[t]
\centering
\caption{Influence of threshold $t$ ($\tau=\ps t$) on the resulting F1 score at $\varepsilon=8$. \# Buc. is the
number of released buckets.}
\label{tab:t}

\begin{tabular}{lccc}
\toprule
 & \multicolumn{3}{c}{$t$} \\
\cmidrule(l){2-4}
 Quantity & $100$ & $200$ & $400$ \\
\midrule
 F1 \Inst{}
& $0.33$
& $0.43$
& $0.52$ \\

F1 \Lmsys{}
& $0.69$
& $0.78$
& $0.73$ \\

\midrule

 \#Buc. \Inst{}
& $89$
& $57$
& $39$ \\

\#Buc. \Lmsys{}
& $424$
& $175$
& $59$ \\

\bottomrule
\end{tabular}
\end{table}

\subsection{More Ablation Studies}
For the ablation study on edge length $L$ shown in Section~\ref{app:ablation}, the noise scale is listed in Table~\ref{tab:L-sigma}.

\vspace{3pt}
\noindent\textbf{Embedding model.} The comparison in Table~\ref{tab:emb} is about transferability of observations, not about
which encoder is ``better''. We can observe that the overall trend of utility-privacy trade-off is similar, precision and recall almost reach the non-DP baseline when $\eps\ge8$. The concrete numbers differ, which is expected, since a different embedding model also changes which infrequent texts the mechanism protects and which frequent texts the mechanism reconstructs. More ablation studies and the details of hyperparameters are in Appendix~\ref{app:parameter-detail}.\looseness=-1

\vspace{3pt}\noindent\textbf{Projection dimension $k$.} Table~\ref{tab:k} presents the F1 score for $k=\{10,15,20,25,30\}$. The best $k$ depends on $\varepsilon$, while in general,  larger $k$ is better for $\varepsilon<\infty$. This is because larger $k$ reduces the failure probability of the sensitivity bound (recall Lemma~\ref{lem:high-prob-sens}), allowing more freedom to choose the noise scale for centroid perturbation. 

\vspace{3pt}\noindent\textbf{Frequency threshold $t$.} Together with $r$, $t$ defines the frequent and infrequent texts from the input (recall Definition~\ref{def:euclidean-secret}) and therefore the evaluation metrics. Hence, the results for different $t$'s are not directly comparable. We also note that $\dfre$ decays exponentially in $\tau=\ps t$,
so at the desired overall $\delta=10^{-6}$, our Fre-E2T is feasible only for
$t\gtrsim 100$ at our subsampling rates. From Table~\ref{tab:t}, as we increase $t$ from $100$ to $200$, the released number of buckets decreases since fewer buckets would satisfy the frequent-text requirement. On the other hand, there is not clear utility trend at different $t$'s. Our main experiment setup selected $t=100$, which is a small number compared with the overall population.

\subsection{More on Text Utility}\label{app:diversity}

Let $p_{\theta}$ be an autoregressive language model and let a text
$w$ tokenize into $(w_1, \dots, w_T)$ under $p_{\theta}$'s tokenizer. The
\emph{per-token negative log-likelihood} (NLL) of $w$ is its average
next-token surprisal in nats,
\begin{equation}
\operatorname{NLL}(w) \;=\; -\frac{1}{T-1} \sum_{t=2}^{T}
\log p_{\theta}\!\left(w_t \mid w_1, \dots, w_{t-1}\right).
\label{eq:nll}
\end{equation}
The perplexity score of $w$ is the exponential of the per-token NLL,
$\operatorname{PPL}(w) = \exp\bigl(\operatorname{NLL}(w)\bigr)$. Lower NLL (equivalently,
lower PPL) means the text is closer to natural language as
modeled by $p_{\theta}$. We use GPT-2 as $p_{\theta}$. 

Here, we report the median of the per-text NLL over each
text set, together with the fraction of texts with
$\operatorname{NLL}(w) \le 6$, on TaylorAI. From Table~\ref{tab:textquality-nll-base}, we can see that for the outputs from Fre-E2T, their NLL approaches that of the real corpus as the privacy
budget increases, matching the results from Section~\ref{sec:text-quality}.\looseness=-1

\begin{table}[h!]
\centering
\caption{NLL of Fre-E2T outputs vs.\ real text of TaylorAI evaluated using under GPT-2. The real texts are a random
1000-text sample of the frequent texts.}
\label{tab:textquality-nll-base}
\begin{tabular}{lrr}
\toprule
Text set & Median NLL & \% texts NLL $\le 6$ \\
\midrule
Fre-E2T ($\varepsilon = 4$)  & 5.28 & 77.9\% \\
Fre-E2T ($\varepsilon = 16$) & 4.29 & 95.4\% \\
Real texts & 3.43 & 99.5\% \\
\bottomrule
\end{tabular}
\end{table}

\subsection{Evaluation on All Input Texts}\label{sec:app-all-text-evalutaion}

\begin{table*}[t]
\centering
\normalsize
\setlength{\tabcolsep}{4pt}
\caption{Utility measured against \emph{all} input texts $V$. Compared with
Table~\ref{tab:combined_dp_results} that reports the same runs against the
frequent subset $V\setminus V_{\mathrm{infre}}$, precision is unchanged
or higher here (the nearest-text condition is evaluated over a larger
set); recall falls by the infrequent mass, which our Fre-E2T mechanism
deliberately does not support.}
\label{tab:alltexts}

\begin{tabular}{l*{12}{c}}
\toprule
& \multicolumn{3}{c}{TaylorAI}
& \multicolumn{3}{c}{Instructions-2M}
& \multicolumn{3}{c}{LMSYS Chat}
& \multicolumn{3}{c}{Yelp} \\
\cmidrule(lr){2-4}
\cmidrule(lr){5-7}
\cmidrule(lr){8-10}
\cmidrule(lr){11-13}

$|V_{\mathrm{infre}}|/|V|$
& \multicolumn{3}{c}{66\%}
& \multicolumn{3}{c}{90\%}
& \multicolumn{3}{c}{70\%}
& \multicolumn{3}{c}{20\%} \\
\midrule

$\varepsilon$
& Prec. & Rec. & F1
& Prec. & Rec. & F1
& Prec. & Rec. & F1
& Prec. & Rec. & F1 \\
\midrule

$4$
&0.31 &0.33 &0.32
&0.27 &0.03 &0.06
&0.32 &0.22 &0.26
&0.30 &0.77 &0.44 \\

$8$
&0.62 &0.35 &0.45
&0.53 &0.03 &0.06
&0.54 &0.26 &0.35
&0.61 &0.83 &0.70 \\

$16$
&0.82 &0.39 &0.53
&0.60 &0.04 &0.07
&0.77 &0.28 &0.41
&0.87 &0.85 &0.86 \\

$\infty$
&0.83 &0.41 &0.55
&0.71 &0.04 &0.07
&0.85 &0.28 &0.42
&0.97 &0.86 &0.91 \\

\bottomrule
\end{tabular}
\end{table*}

Recall Eq.~\ref{eq:prec} to Eq.~\ref{eq:l2}. We restrict the utility metrics to semantically frequent texts. This is a design consequence, not a favorable selection: semantic support protection (Definition~\ref{def:semantic-protect}) requires that infrequent texts remain far from the generated texts. Table~\ref{tab:alltexts} repeats every run of the main Table~\ref{tab:combined_dp_results} and evaluate against the full input set $V$. Precision is unchanged or higher, since a synthetic text now only has to be close to \emph{some} input text. Recall falls by the infrequent mass. For example, at $\varepsilon=\infty$ it is $0.35$ for frequent AOL texts (recall Table~\ref{tab:combined_dp_results}) and $0.03$ over all of them, since in AOL $92\%$ of the corpus is infrequent. On TaylorAI, the recall decreases from $0.97$ to $0.41$ with $66\%$ infrequent input texts. The gap between the two tables exactly that proves that our protection for infrequent texts works.

\section{Details of Distributed \name}\label{app:alg-detail}

We refer to Figure~\ref{fig:illustration} for the distributed implementation of Fre-E2T and Figure~\ref{fig:grid-illustration} for the illustration of the partitioning procedure on the $k$-dimensional projection space. 

We refer to Algorithm~\ref{alg:hash} for the verifiable oblivious pseudorandom function for randomly permuting user indices. Truncated Shited Laplace noise is introduced to enforce DP for non-frequent user texts during the heavy hitter identification process. We refer to Definition~\ref{def:tsd} and Algorithm~\ref{alg:create-dummy} for the definition and the algorithmic process, respectively.

\begin{defn}[Truncated Shifted Discrete Laplace (TSDLap)]\label{def:tsd}
The distribution of TSDLap on $\{0, \ldots, 2\gamma\}$ is defined as
\begin{align*}
f_{\text{TSDLap}(\lambda, \gamma)}(u) =
\begin{cases}
\dfrac{\exp\!\left(-\dfrac{|u - \gamma|}{\lambda}\right)}{A}, & \text{if } u \in \{0, \ldots, 2\gamma\},\nonumber\\
0, & \text{otherwise,}
\end{cases}
\end{align*}
where $\lambda \in (0,1)$ is the scale, $\gamma$ is the shift from $0$, and
\[
A = \sum_{u=0}^{2\gamma} \exp\!\left(-\dfrac{|u - \gamma|}{\lambda}\right)
   = 1 + 2\sum_{u=1}^{\gamma} \exp\!\left(-\dfrac{u}{\lambda}\right).
\]
\end{defn}

We note that \Ssyn{} sees the multiset of sub-threshold tags with their exact multiplicities; the tags themselves are
pseudorandom to any party without the OPRF key, and it is the resulting
count-of-counts---how many tags were seen with each multiplicity---that Algorithm~\ref{alg:create-dummy}'s TSDLap padding makes differentially private~\cite{nebula}. 

We also provide the full description of the distributed discrete Gaussian mechanism to obtain the aggregated vector sum of distributed users~\cite{DiscreteGaussian} (other mechanisms such as~\cite{sk,pbinomial} lead to similar results). See Algorithm~\ref{alg:aclient} and Algorithm~\ref{alg:aserver}. In short, each user locally quantizes its vector, perturbs it with discrete Gaussian noise, and then takes the modulo in order to prepare for secure aggregation~\cite{secagg,samplablefa,SecureML}. After secure aggregation, the untrusted reconstructs an estimate for the sum by unwinding the modulo operation and the quantization. When the quantization granularity is fine, e.g., $\beta=2^{-16}$, the resulting data utility (error compared with the original, non-quantized, non-DP vector sum) is almost the same as perturbing the sum directly with continuous Gaussian noise in the centralized setting~\cite{DiscreteGaussian,abadi}. In our simulation of \name, we adopt the continuous Gaussian without loss of generality. 

\begin{theorem}[Distributed discrete Gaussian (complete version)~\cite{DiscreteGaussian}]\label{thm:main-dgg-formal}
Let $\phi$ be the parameter controlling the rounding bias, $\beta>0$ be the quantization parameter, $\sigma$ be the parameter of the local discrete Gaussian noise. Under the same assumption as Theorem~\ref{thm:main-gaussian}, define
\begin{align}
\Delta_2^2 &:= \min \Big(\big(\Delta_{\textit{dist}} + \beta \sqrt{d}\big)^2\,,\\
&\Delta_{\textit{dist}}^2 + \frac{\beta^2 d}{4} + \sqrt{2 \log \left(\frac{1}{\phi}\right)} \cdot \beta \cdot \big(\Delta_{\textit{dist}} + \frac{\beta}{2} \sqrt{d}\big)\Big)\,,\nonumber
\end{align}
\begin{align}
\kappa &= 10 \cdot \sum_{j=1}^{\tau-1} \exp\big(- \frac{2\pi^2 \sigma^2}{\beta^2} \cdot \frac{j}{j+1}\big)\,,\\
\rho_{agg} &:= \min \Big(
  \frac{\Delta_2^2}{2\tau \sigma^2} +  \kappa d, \;
  \frac{1}{2}\big(\frac{\Delta_2}{\sqrt{\tau} \sigma} + \kappa\sqrt{d}\big)^2
\Big).
\end{align}
The outcome of the perturbed sum of embedding vectors from Alg.~\ref{alg:all} satisfies $\rho_{agg}$-CDP, and $(\varepsilon_{agg},\delta_{agg})$-DP with \begin{align}
    \delta_{agg} = \inf_{\alpha \in (1,\infty)} \frac{\exp((\alpha-1)(\alpha\rho_{agg}-\varepsilon_{agg}))}{\alpha - 1} \cdot \left( 1 - \frac{1}{\alpha} \right)^{\alpha}.
\end{align}
\end{theorem}

\vspace{1pt}
\noindent\textbf{Connection with~\cite{pmlrv139chang21a}.} Our clustering process of projecting data into a lower-dimension space and then partitioning the data into buckets is in spirit similar to the DP clustering work~\cite{pmlrv139chang21a}. Our difference is in the design of the distributed algorithm. While ~\citet{pmlrv139chang21a} focuses on the local DP model, we employ the more recent sample-and-threshold privacy framework~\cite{sandtdp} which has an efficient distributed implementation~\cite{nebula} with the help of an untrusted randomness server. In addition, we also exploit the structure of $\mathcal{L}_2$ distance-based clustering to reduce the sensitivity when computing the cluster centroid.

\vspace{3pt}
\noindent\textbf{On Sample-and-threshold Histogram.} Algorithm~\ref{alg:centralized} releases two things for each heavy bucket $h$: the pair $(h, n_h)$ in
the heavy-hitter step (line~15), and the noisy centroid $\hat{\mu}_h$ in the
aggregation step (line~19). Here $n_h$ is \emph{not} the number of users in
bucket $h$: it is the number of users that \emph{sampled themselves into} the
release, i.e.\
$n_h \sim \mathrm{Bin}(N_h, p_s)$ where $N_h$ is the true bucket size. The pair $(h, n_h)$ is precisely the output
that the sample-and-threshold guarantee of Lemma~\ref{lem:s-and-t-dp} certifies with
$(\varepsilon_{fre}, \delta_{fre})$. The raw count $N_h$ itself is never released. The centroid step then divides the noise-perturbed sum by the
\emph{already released} $n_h$ (line~18); since both the numerator's DP guarantee
(Theorem~\ref{thm:main-gaussian}) and the divisor's (Lemma~\ref{lem:s-and-t-dp}) are already
accounted in Theorem~\ref{thm:dp-final-fre-e2t} by linear composition, this division is post-processing and
incurs no additional privacy cost~\cite{dpbook}.

\subsection{Deployment Costs}
\label{app:costs}

\begin{table}[h]\centering\small
\caption{Per-user cost of distributed Fre-E2T versus a distributed Aug-PE.}
\label{tab:vsaugpe}
\begin{tabular}{@{}lrrr@{}}
\toprule
Per user & Fre-E2T & Aug-PE & ratio\\
\midrule
Computation          & $0.40$\,ms (once) & $55.9$\,ms ($20$ rounds) & $140\times$\\
Upload               & $8.3$\,kB  & $1.68$\,MB & $2.0\times10^{2}$\\
Download             & $6.1$\,kB  & $1.29$\,GB & $2.1\times10^{5}$\\
\bottomrule
\end{tabular}
\end{table}

We complement the cost analysis in Section~\ref{sec:agg} with concrete numbers. The costs are measured on an AMD EPYC~9554 (124 cores) with one
NVIDIA H100 80\,GB without real test on distributed environments. Parameters are $\varepsilon=2$, $k=20$, $\delta=10^{-6}$, $d=768$,
$d'=1024=O(d)$, and the reported dataset is TaylorAI.

The per-operation costs in the distributed Fre-E2T are: one VOPRF evaluation $89.5\,\mu$s; tag counting $1.05\,\mu$s per tag; share
accumulation $10.7\,\mu$s per share; bucket reconstruction $60.5\,\mu$s per bucket; and text embedding $0.29$\,ms per text on the GPU amortized at batch size $256$. A user spends around
$0.4$\,ms online in total, uploads $8.3$\,kB and downloads $6.1$\,kB. Both servers finish in a couple of minutes for $N$ VOPRF evaluations and the computation for the centroids. There is no need for
GPU for the online phases. The offline vec2text inversion takes $133$ms per text. The dummy traffic that protects sub-threshold buckets
is around $1710$ messages and $113$\,kB for the entire run.

\vspace{0pt}
\noindent\textbf{Comparison with Aug-PE.}
Aug-PE~\cite{xie2024differentially} keeps a population of $N_{syn}$ synthetic texts and, in each
of $T$ iterations, generates $L{-}1$ rephrasings of every kept sample, embeds the
resulting pool of $N_{cand}=L\,N_{syn}$ candidates, has \emph{every} private
record vote for its nearest candidate, perturbs that histogram and keeps the top
$N_{syn}$. Its per-iteration cost is thus $(L{-}1)N_{syn}$ generations,
$N_{cand}$ embeddings and $\Theta(N_{pri}N_{cand}d)$ for computing the embedding distance, repeated $T$
times. In addition, because the vote is over private data, the distance computation must be done
on the users' devices, which must download all candidate text embeddings first. With $N_{syn}=3000$, $L=7$, $T=20$ (i.e.\ $21{,}000$ candidates scored per
iteration), one iteration of Aug-PE takes $64.6$\,s on an H100 GPU. A full run takes $21.7$ minutes that require users' synchronous updates and continuous participation, much longer than Fre-E2T.

\begin{figure*}[t!]
\centering
\includegraphics[width=0.99\textwidth]{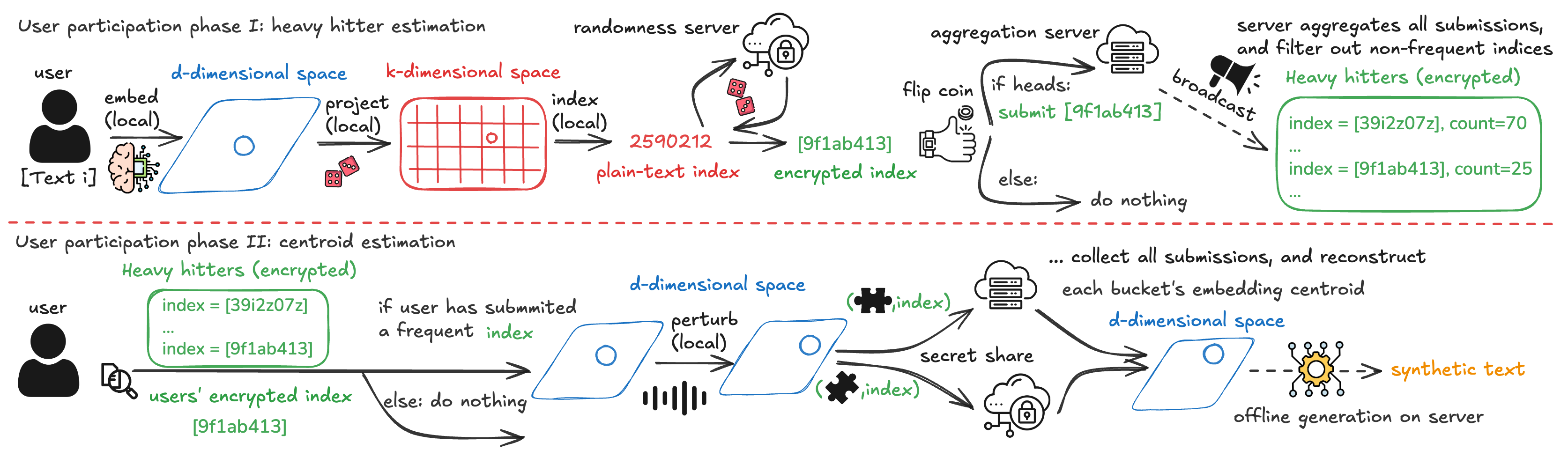} 
\caption{Overall idea of our algorithm in the distributed setting.}
\Description{A two-phase distributed protocol. Users first embed, project, and tag local texts and probabilistically submit encrypted indices for heavy-hitter estimation. Users associated with frequent indices then perturb and secret-share their embeddings. Servers aggregate shares to reconstruct private centroids for offline text synthesis.}
\label{fig:illustration}
\end{figure*}

\begin{figure}[t]
\centering
\includegraphics[width=\linewidth]{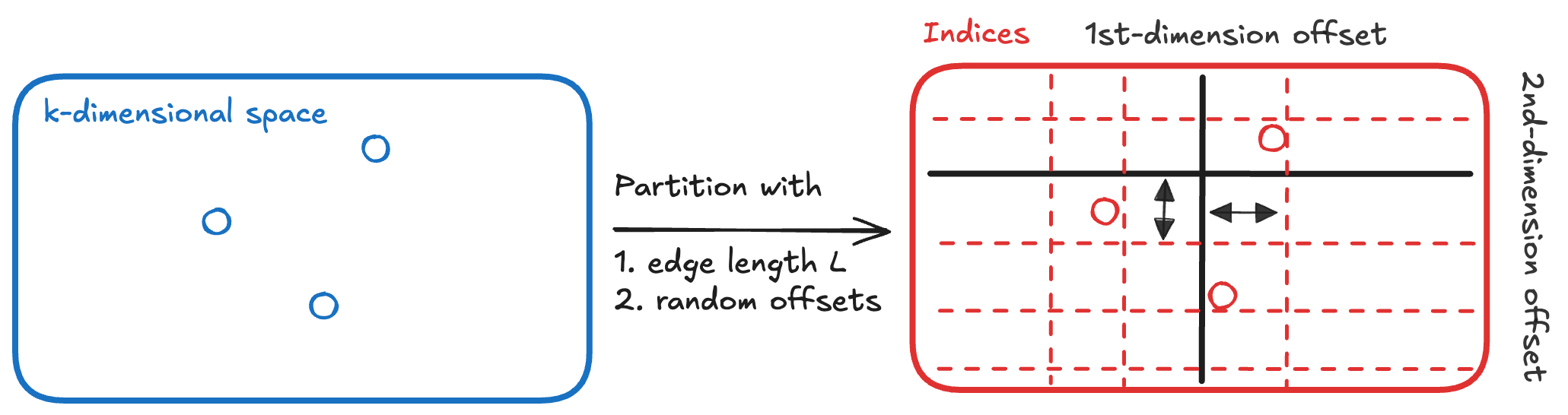} 
\caption{Partitioning the $k$-dimensional space with a fixed edge length $L$ and random offsets.}
\Description{Points in an embedding projection are assigned to a regular grid with edge length L. Independent random offsets shift the grid boundaries along each dimension; a two-dimensional example illustrates the construction.}
\label{fig:grid-illustration}
\end{figure}

On the server side, neither method is expensive in computation or communication. The
difference is on the user side: Fre-E2T's
user cost is paid once and is independent of $T$ and of $N_{cand}$, whereas
Aug-PE's recurs every iteration and scales with the candidate pool. We highlight this difference in Table~\ref{tab:vsaugpe}. 

\begin{algorithm}[t!]
\DontPrintSemicolon
\SetAlgoLined
\KwIn{user $i$ holding $b_i\in\mathbb{Z}^{k}$; tagging server; VOPRF secret key $\mathsf{msk}$; hash function $g(\cdot)$.}
\KwOut{Random string $h_i$}
\textbf{user $i$:} Hash the index $u \leftarrow g(b_i)$ \;
\textbf{user $i$:} Sample blinding $r' \gets_R \mathbb{Z}_q^*$\;
\textbf{user $i$:} Compute $\textsf{blind} \leftarrow u^{\,r'}$ \;
\textbf{user $i$:} Send blinded hash $\textrm{blind}$ to the server\;
\textbf{Tagging server:} Compute $z \leftarrow \textsf{blind}^{\,\mathsf{msk}}$ \;
\textbf{user $i$:} Unblind the response $w \leftarrow z^{\,1/r'}$\;
\textbf{user $i$:} Obtain  $h_i \leftarrow g \big(w, b_i\big)$ \;
\textbf{Return} $h_i$ \;

\caption{\textit{Client-TaggingServer}}
\label{alg:hash}
\end{algorithm}

\begin{algorithm}[t!]
\DontPrintSemicolon
\SetAlgoLined
\KwIn{A public thresholding value $\tau$; Truncated Shifted Discrete Laplace distribution $\textsf{TSDLap}(\cdot)$ with scale $\lambda$ and shift $\gamma$.}
\KwOut{A list of dummy indices}
Dummy $\leftarrow \{\}$\;
\For{$n = 1, \ldots, \tau - 1$}{
   Sample $Q \leftarrow \textsf{TSDLap}(\lambda,\gamma)$\;
    $\{\textsf{tag}_j\}_{j=1}^{Q} \leftarrow \textsf{UniqueIndexGenerator}(\alpha)$ \;
    \For{$j = 1, \ldots, Q$}{
        Append $n$ copies of $\textsf{tag}_j$ to Dummy\;
    }
}
\Return Dummy\;
\caption{\textit{Dummy-Data-Creation}}
\label{alg:create-dummy}
\end{algorithm}

\begin{algorithm*}[t]
\DontPrintSemicolon
\caption{User Procedure $\mathcal{A}_{\text{client}}$ for distributed discrete Gaussian mechanism~\cite{DiscreteGaussian}}
\label{alg:aclient}
\SetKwInput{Input}{Input}
\SetKwInput{Params}{Parameters}
\SetKwInput{Output}{Output}

\Input{Vector $x_i \in \mathbb{R}^d$.}
\Params{Dimension $d$; $\mathcal{L}_2$ sensitivity $\Delta_{dist}$; quantization granularity $\beta>0$; modulus $m\in\mathbb{N}$; noise scale $\sigma>0$; bias $\phi\in[0,1)$; a uniform random sign vector $\xi \in \{-1,+1\}^d$ and a Walsh--Hadamard matrix $H\in\{-1/\sqrt{d},+1/\sqrt{d}\}^{d\times d}$ that are shared among all parties.}
\Output{$z_i \in \mathbb{Z}_m^d$ for the secure aggregation protocol.}

\BlankLine
$x'_i \gets \dfrac{1}{\beta}\cdot x_i \in \mathbb{R}^d$.\;

$x''_i \gets H D_\xi x'_i \in \mathbb{R}^d$.\;

\Repeat{}{
  Draw $\tilde{x}_i \in \mathbb{Z}^d$ from a product distribution such that
  $\mathbb{E}[\tilde{x}_i]=x''_i$ and $\lVert \tilde{x}_i - x''_i \rVert_\infty < 1$.\;

}{
  $\lVert \tilde{x}_i \rVert_2 \le
  \min\!\left\{\, \dfrac{\Delta_{\textit{dist}}}{\beta}+\sqrt{d},\;
  \sqrt{\,\dfrac{\Delta_{\textit{dist}}^2}{\beta^2} + \dfrac{1}{4}d + \sqrt{2\log(1/\phi)}\cdot\!\left(\dfrac{\Delta_{\textit{dist}}}{\beta} + \dfrac{1}{2}\sqrt{d}\right)}\;\right\}$\;
}

Sample $a_i \in \mathbb{Z}^d$ with independent coordinates
$a_{i,j}\sim \mathcal{N}_{\mathbb{Z}}\!\left(0,\sigma^2/\beta^2\right)$ for $j=1,\dots,d$.\;
$z_i \gets (\tilde{x}_i + y_i) \bmod m$.\;
\end{algorithm*}

\begin{algorithm*}[t!]
\DontPrintSemicolon
\caption{Server Procedure $\mathcal{A}_{\text{server}}$~\cite{DiscreteGaussian}}
\label{alg:aserver}

\SetKwInput{Input}{Input}
\SetKwInput{Params}{Parameters}
\SetKwInput{Output}{Output}

\Input{Aggregated modulo sum $\bar z \in \mathbb{Z}_m^{d}$ obtained via secure aggregation.}
\Params{Same as Algorithm.~\ref{alg:aclient}.}
\Output{$y\in \mathbb{R}^{d}$}

Map $\mathbb{Z}_m \to \{-\lfloor m/2\rfloor,\ldots,0,\ldots,\lfloor m/2\rfloor\}$ coordinate-wise so that
\(
\bar z' \in [-m/2,m/2]^d \cap \mathbb{Z}^d
\quad\text{and}\quad
\bar z' \bmod m = \bar z .
\)\;
Compute $y \gets \beta\, D_{\xi}\, H^\top \bar z'$, where $D_{\xi}$ is diagonal with $\xi$ on the diagonal.\;
\end{algorithm*}

\section{Proof for Semantic Support Protection}\label{app:proof-of-semantic-support}

We prove Theorem~\ref{thm:proof-semantic-all} in three steps. First, we argue that for any infrequent text $x^*$ that has at most $n^*$ data points that is $r$-close to itself, after subsampling, not too many of its close neighbors in $X$ will end up in the same cluster, using the classic concentration inequality, and that the cluster will have a small diameter (using the results established in Sec.~\ref{sec:privacy-ddg}). Next, we focus on a single cluster, and argue that conditioned on there are not too many close neighbors to $x^*$ in it, the cluster centroid will be far from $x^*$. Finally, we show that the DP noise will push the released centroid further away $x^*$, with a high probability. 

\vspace{1pt}
\noindent\textbf{Good event $\mathcal{B}$.} We focus on any cluster $B$ with index $b$ in the $k$-dimensional space that has accumulated at least $h$ data points. Note that $h$ was released in the heavy hitter estimation step with DP. So here, using $h$ incur no additional privacy cost. 

We denote the corresponding $d$-dimensional embedding vectors as $X_{B}=\{x_i\,\mid\, h_i=b\}$. We define event $\mathcal{B}$ as ``Within bucket $B$, there are at most $2p_sn^*$ data points that are $r$ close to $x^*$ \textit{and} every two data points have distance at most $r$''. Lemma~\ref{lem:few-in-bucket} follows from the multiplicative Chernoff bound~\cite{MitzenmacherUpfal2005} and Lemma~\ref{lem:high-prob-sens}.

\begin{lem}[Event $\mathcal{B}$ almost always happens]\label{lem:few-in-bucket}
Within bucket $B$ of count $h$, there are at most $2p_sn^*$ data points that are $r$ close to $x^*$ and every pair of data points are within distance at most $r$. This happens with probability at least 
\begin{align}\label{eq:few-in-bucket}
1-\exp(-0.38 p_sn^*)-\binom{h}{2}\Big(f\Big(\frac{\sqrt{k} L}{r}\Big)\Big)^k,
\end{align}
with $f$ defined as in Eq.~\ref{eq:def-f-centralized}.   
\end{lem}

\vspace{1pt}
\noindent\textbf{Distance between a centroid and $x^*$.} Conditioned event $\mathcal{B}$, we show that no bucket centroid is close to $x^*$, without DP noise. Recall that each bucket contains $h$ data points (with $h\ge\tau$), so at least $h-2p_sn^*$ data points in the bucket is at least $r$-away from $x^*$. We take any such $x_{\textit{far}}$, then $\|x_{\textit{far}}-x^*\|_2\geq r$. We define unit vector \(
    u:= \frac{x_{\textit{far}}-x^*}{\|x_{\textit{far}}-x^*\|_2}.\)
Now consider any other data point that is far from $x^*$, written as $x_{\textit{far}}'$. Let $\theta$ be the angle between $x_{\textit{far}}-x^*$ and $x'_{\textit{far}}-x^*$. Then we can obtain 
\begin{align*}
\cos \theta=\frac{\|x_{\textit{far}}-x^*\|^2_2+\|x'_{\textit{far}}-x^*\|^2_2-\|x_{\textit{far}}-x'_{\textit{far}}\|^2_2}{2\|x_{\textit{far}}-x^*\|_2\|x'_{\textit{far}}-x^*\|_2}\geq \frac{1}{2},    \end{align*} since and $\|x_{\textit{far}}-x'_{\textit{far}}\|_2\le r\le \min\{\|x_{\textit{far}}-x^*\|_2,\|x'_{\textit{far}}-x^*\|_2\}$ by construction. Therefore for any $x'_{\textit{far}}$, we have \(
u^T(x'_{\textit{far}}-x^*)= r\cos\theta  \geq\frac{r}{2}\). For any $x_{near}$ such that $\|x_{near}-x^*\|\leq r$, \(u^T(x_{near}-x^*)\ge -r\) holds automatically. 
Since there are at least $h-2p_sn^*$ far data points, the centroid of $B$, written as $\mu(B)$, must satisfy
\begin{align}
   u^T(\mu(B)-x^*)\ge \frac{h-6p_sn^*}{2h} r. 
\end{align}
Since $h\ge\tau$, \( \|\mu(B)-x^*\|_2\geq \frac{\tau-6p_sn^*}{2\tau}r.\) This means that the distance from $x^*$ to $\mu(B)$ is dominated by the majority far points.

\vspace{1pt}
\noindent\textbf{Effect of DP noise.} We consider injecting $G\sim\mathcal N(0,(\sigma^*)^2)$ into the centroid $\mu(B)$. We define $\widetilde x \;=\; \mu(B) + G$, and $R=\widetilde x-x^*$. 

\begin{align}
&\Pr[\|\widetilde x -x^*\|_2\leq r] = (2\pi(\sigma^*)^2)^{-d/2}\cdot\nonumber\\
&\qquad \int_{\|R\|_2\leq r}  \exp\big(-\frac{\|R-(\mu(B)- x^*)\|_2^2}{\sigma^2}\big)dR\label{eq:integral}
\end{align}
Let $R'$ follow the distributed Gaussian distribution of mean $\|\mu(B)- x^*\|_2 \mathbf{e}_1$ with covariance $(\sigma^*)^2 \mathbf{I}_d$. By rotation invariance,~\eqref{eq:integral} can be computed as \begin{align}\label{eq:prob-dist}
\hspace{-2mm}\textstyle\Pr\big[\|R'\|_2\leq r\big]=\Pr\Big[\|\frac{R'}{\sigma^*}\|_2\leq \frac{r^2}{(\sigma^*)^2}\Big]=F_{\chi^2_d(\Lambda)}\Big(\frac{r^2}{(\sigma^*)^2}\Big),
\end{align} 
where $F$ is the cumulative distribution function of non-central chi-squared distribution and $\Lambda=\frac{\tau-6p_sn^*}{2\tau}r$. Eq.~\ref{eq:prob-dist} is decreasing in $\Lambda$; hence, it is decreasing in $\|\mu(B)-x^*\|_2$. This in turn motivates Lemma~\ref{lem:few-in-bucket}, where have we bounded $\|\mu(B)-x^*\|_2$ from below. 

The overall failure probability follows from a union bound over Eq.~\ref{eq:few-in-bucket} and Eq.~\ref{eq:prob-dist}. Plugging in $\sigma^*=\frac{\sigma}{h}$, where $h$ is the subsampled count of items landed in the heavy hitter cluster of $x^*$ (namely, the denominator for the centroid computation), we obtain Theorem~\ref{thm:proof-semantic-all}. We note that $h$ is already released through the differentially private heavy hitter histogram perturbed with the random subsampling noises; hence, using $h$ incurs no additional privacy cost.

To convert Theorem~\ref{thm:proof-semantic-all} to a meaningful bound on the text-level semantic protection, one also needs to quantify the error during inverting the vector to a text, namely, the $\mathcal{L}_2$ distance $\|\mu(B)-E(\tilde v)\|$, where $\tilde v$ is the text inverted from centroid $\mu(B)$. This can be difficult in general. However, based on our empirical observations from Figure~\ref{fig:semantic}, the error is not large and text-level SSP is preserved.

\vspace{1pt}
\noindent\textbf{Remark on distributed implementation.} Real-world distributed implementation of securely aggregating DP noises relies on integer-valued implementations, such as
\emph{discrete Gaussian}~\cite{DiscreteGaussian,dgcanonne} (used in this paper),  or \emph{symmetric Skellam}~\cite{sk} (difference of two Poisson variables).
In both cases, for moderate noise levels or in moderately high dimension, as is the case in our problem, the distribution of the squared distance $\|\widehat \mu+\text{discrete noise}-x^*\|_2^2$
can still be approximated by a noncentral chi-squared distribution with the same degrees of freedom $d$
and the almost the same noncentrality $\lambda=r^2/v$, where $v$ is the per-coordinate variance
of the mechanism. 

Now, we consider Algorithm~\ref{alg:all}, where each user independently contributes discrete Gaussian noise $\mathcal{N}_{\mathbb{Z}}(0,\sigma_{disc}^2)$ to the discretized local vector. The sum of $h$ independent discrete Gaussian noises of variance $\sigma_{disc}^2$ is close to a single discrete Gaussian noise of variance $h\sigma_{disc}^2$. Normalizing by $\beta/h$ (here $\beta$ is the discretization parameter with $\beta\ll 1$), the overall variance in the resulting centroid is $\frac{\beta^2\sigma_{disc}^2}{h}$. Hence, the conclusion of semantic support protection for the distributed implementation of \name, Algorithm~\ref{alg:all}, is similar to what we have derived above for a continuous Gaussian, with an algebraic substituion, presented as follows.

For any $x^*\in X$ such that there are at most $t$ data points from $X$ that has distance at most $r$ to $x^*$ in $X$ and any DP centroid $\widehat \mu$ with $h$ items subsampled into the cluster, Algorithm~\ref{alg:all} outputs the centroid $\widehat \mu$ such that  $\|\widehat \mu - x^*\|_2 < \psi r$ with probability 
\begin{align}
\delta^* \lesssim
\binom{h}{2}\Big(f\big(\frac{\sqrt{k}L}{r}\big)\Big)^k + e^{-0.38p_s t} + F_{\chi^2_d(\Lambda)}\Big(\frac{\psi^2 r^2h}{\beta^2\sigma_{disc}^2}\Big).
\end{align}

\section{Extensions}\label{app:ext}

\subsubsection{Each User Has Multiple Data Points}\label{app:ext-more-datapoints}
Previously, we have considered the setting in which each user has exactly $1$ data point. Now we discuss the extension when each user has more than $1$ data points. 

\vspace{0pt}
\noindent\textbf{Naive solution via composition.} During the heavy-hitter estimation, we can introduce a public threshold to restrict the number of buckets that each user can contribute to, say $c$. For each bucket, we make each user contribute at most one of its embedding vectors. The overall privacy guarantee then follows from the composition argument, we can directly multiply the overall $\varepsilon$ and $\delta$ by $c$. Finding the optimal threshold $c$ could be an interesting future work direction~\cite{truncate1,truncate2}

\vspace{0pt}
\noindent\textbf{Naive solution via black-box group DP.} The above naive solution forbids one user to submit to the same bucket index multiple times. This may hurt utility considering that some queries (associated with certain bucket indices are more common). Let's say each user has at most $m$ data points. A straight-forward group DP argument gives $(m\cdot\varepsilon,\frac{e^{m\varepsilon}-1}{e^\varepsilon-1}\cdot\delta)$-DP, where $\varepsilon$ and $\delta$ are the original privacy parameters of the algorithm when each user has exactly one data point. When $\varepsilon$ is not close to $0$, the factor of $\frac{e^{m\varepsilon}-1}{e^\varepsilon-1}$ can be very large. 

\vspace{0pt}
\noindent\textbf{A bespoke  analysis.} Next, we discuss a more bespoke analysis. Let the group size be $l$ (originally, we have considered $l=1$). We first focus on the heavy hitter estimation phase. For now, we consider two databases consisting of copies of the same index. Dataset $X$ has $m$ copies while dataset $X'$ has $m-l$ copies. The sampled count $V$ and $V'$ for $X$ and $X'$ then follows the binomial distribution $Bin(m,p_s)$ and $Bin(m-l,p_s)$, respectively. 
\begin{align}
    \frac{\Pr[V'=v]}{\Pr[V=v]}=\frac{\binom{m}{v}p_s^v(1-p_s)^{m-v}}{\binom{m-l}{v}p_s^v(1-p_s)^{m-l-v}},
\end{align}
for any $v\in\mathbb{N}\,, v\le m-l$. We can obtain
\begin{align}\label{eq:group-dp-first-step}
    \frac{\Pr[V'=v]}{\Pr[V=v]}=(1-p_s)^l\Pi_{j=0}^{l-1}\frac{m-j}{m-v-j},
\end{align}
which collapses to $(1-p_s)^l \frac{m}{m-v}$ when $l=1$, i.e., the expression that appears in the original~\cite{sandtdp}. We want Eq.~\ref{eq:group-dp-first-step} to be bounded within $[e^{-l\cdot\varepsilon_{HH}},e^{l\cdot\varepsilon_{HH}}]$, with a high probability. The lower bound is easy, it suffices to make $1-p_s\ge e^{-\varepsilon_{HH}}$. For the upper bound, it suffices to make sure that 
\begin{align}
    (1-p_s) \frac{m-j}{m-v-j} \le e^{\varepsilon_{HH}},
\end{align}
for each $j=0,1,...,l-1$. It suffices to ensure that 
\begin{align}
    (1-p_s) \frac{m-l+1}{m-v-l+1} \le e^{\varepsilon_{HH}}
\end{align}
holds, which is equivalent to
\begin{align}
    v\le q(m-l+1),
\end{align} 
where we have set $q:=1-(1-p_s)e^{-\varepsilon_{HH}}$.
Since $v$ could take value in $0,1,...,m$, the $\delta_{l}$ can be written as
\begin{align}
    \delta_{l}=\Pr[Bin(m,p_s)\ge \max\Big(\tau, q(m-l+1)\Big)]
\end{align}
When $l=1$, $
\delta_{1}=\Pr[Bin(m,p_s)\ge \max(\tau, qm)],$ which is the same as in~\cite{sandtdp}. As a sanity check, we note that $\delta_l$ is increasing in $l$---i.e., it is more likely for the binomial variable to be lower bounded by a smaller value.

There are two cases to consider. $q(m-l+1)\le\tau$ or $q(m-l+1)\ge\tau$; equivalently, $m\le \frac{\tau}{q}+l-1$ or $m\ge \frac{\tau}{q}+l-1$

\vspace{0pt}
\noindent\textbf{Case I: $m\le \frac{\tau}{q}+l-1$.} We can upper bound $\Pr[Bin(m,p_s)\ge \tau]$ by $\Pr[Bin(\frac{\tau}{q}+l-1,p_s)\ge \tau]$.

\vspace{0pt}
\noindent\textbf{Case II: $m\ge \frac{\tau}{q}+l-1$.} In this case, we upper bound $\Pr[Bin(m,p_s)\ge qm]$. We upper bound this probability by the Chernoff bound, that is, the probability that at least $qm$ out of $m$ independent trials are successful, where each trial succeeds with probability $p_s<q$. In expectation, the number of success is $pm<qm$, therefore, the Chernoff bound error decreases as $m$ increases---when $m$ is larger, it becomes more difficult for the off-mean event to happen. Therefore, it suffices to consider the boundary case $m=\frac{\tau}{q}+l-1$, which also happens to be the worst case above. We obtain 
\begin{align}
\delta_{l}&=\Pr[Bin(\frac{\tau}{q}+l-1,p_s)\ge q\cdot (\frac{\tau}{q}+l-1)]\\
&\le \Pr[Bin(\frac{\tau}{q}+l-1,p_s)\ge \tau]\\
&\le \exp\!\left(-m_l\,D\!\left(\frac{\tau}{m_l}\Big\|p_s\right)\right)\label{eq:delta_l},
\end{align}
with $m_l=\frac{\tau}{q}+l-1$, and $D$ stands for the KL divergence between two Bernoulli trials. The first inequality holds since $q(l-1)\ge 0$ (this formalization unifies \textbf{Case II} and \textbf{Case I}); the second inequality is obtained from the Markov's inequality and then solving for the optimal order, similar to~\cite{sandtdp}. 

In summary, when two datasets differ by $l$ copies of the same index, the overall $\varepsilon_{HH}$ for heavy hitter estimation multiplies by $l$ while the overall $\delta_{HH}$ can be computed via the union bound:
\begin{align}
 \delta_{HH} =   \max\sum_{\{\Delta_j\}_{j=1}^l} \delta_j,
\end{align}
where $\delta_j$ is computed via Eq.~\ref{eq:delta_l}, and the maximum is taken over all possible sets of $\{\Delta_j\}_{j=1}^l$ such that each $\Delta_j\in\mathbb{N}$ and that $\sum_{j=1}^l \Delta_j=l$. In particular, if all data points of a user are the same, then exactly one $\Delta_j=l$, and the rest are $0$.

Next, we discuss the privacy guarantee for the DP centroid computation. When all $l$ data points are in one cluster, the sensitivity multiplies by $l$, so does the failure probability (since each data in the bucket point may fail to stay close to the different data point in the same bucket). We can apply Lemma~\ref{lem:high-prob-sens} and Theorem~\ref{thm:main-gaussian} with the sensitivity multiplied by $l$ to obtain $\varepsilon_{agg,l}$ and $\delta_{agg,l}$ for any group size $l$.

The overall parameters are expressed as 
\begin{align}
 \varepsilon_{agg} &=   \max\sum_{\{\Delta_j\}_{j=1}^l} \varepsilon_{agg,j}\,,\\
  \delta_{agg} &=   \max\sum_{\{\Delta_j\}_{j=1}^l} \delta_{agg,j}
\end{align}
where the maximum is taken over all possible sets of $\{\Delta_j\}_{j=1}^l$ such that each $\Delta_j\in\mathbb{N}$ and that $\sum_{j=1}^l \Delta_j=l$. In particular, if all data points of a user are the same, then exactly one $\Delta_j=l$, and the rest are $0$.

\subsubsection{Bounded Differential Privacy}\label{app:ext-bounded-dp}
Our analysis in Section~\ref{sec:privacy} focuses on the unbounded DP model~\cite{icalpdp}, where two datasets are called neighboring if one can be obtained from the other by adding or removing exactly one record in the dataset. Now we discuss the privacy guarantee under the bounded DP model~\cite{dmns}, where two datasets are called neighboring if one can be obtained from the other by altering exactly one record in the dataset (this also captures the scenario when one user's record is zeroed out). We provide the full argument for completeness. 

We say two datasets $D,D'\in\mathcal{D}$ are \emph{unbounded-adjacent}, denoted $D \sim_u D'$, if $D'$ can be obtained from $D$ by adding or removing exactly one record; i.e., there exists $x\in\mathcal{X}$ such that either
$D' = D \cup \{x\}$ or $D = D' \cup \{x\}$. Accordingly, we say two datasets $D,D'\in\mathcal{D}$ are \emph{bounded-adjacent}, denoted $D \sim_b D'$, if they have the same size and differ in exactly one record.
Equivalently, there exist a dataset $C$ and records $x,x'\in\mathcal{X}$ such that
\(
D = C \cup \{x\}, D' = C \cup \{x'\}.\)

\begin{lem}[Unbounded DP to bounded DP parameter conversion~\cite{dpbook}]
If a mechanism $M$ is $(\varepsilon,\delta)$-DP with respect to unbounded adjacency $\sim_u$,
then $M$ is $\bigl(2\varepsilon,\,(1+e^{\varepsilon})\delta\bigr)$-DP with respect to bounded adjacency $\sim_b$.
\end{lem}

\begin{proof}
Assume $M$ is $(\varepsilon,\delta)$-DP w.r.t.\ $\sim_u$.
Fix any pair $D \sim_b D'$. By the definition of bounded adjacency, there exist a dataset $C$ and records $x,x'\in\mathcal{X}$
such that $D = C\cup\{x\}$ and $D' = C\cup\{x'\}$. Note that $D \sim_u C$ (remove $x$) and $C \sim_u D'$ (add $x'$).
Let $S\subseteq \mathcal{Y}$ be any measurable set. Applying $(\varepsilon,\delta)$-DP w.r.t.\ $\sim_u$ to the pair $(D,C)$ gives
\begin{equation}
\Pr[M(D)\in S] \;\le\; e^{\varepsilon}\Pr[M(C)\in S] + \delta. \label{eq:step1}
\end{equation}
Applying $(\varepsilon,\delta)$-DP w.r.t.\ $\sim_u$ to the pair $(C,D')$ gives
\begin{equation}
\Pr[M(C)\in S] \;\le\; e^{\varepsilon}\Pr[M(D')\in S] + \delta. \label{eq:step2}
\end{equation}
Substituting Eq.~\ref{eq:step2} into Eq.~\ref{eq:step1} yields
\begin{align}
\Pr[M(D)\in S]
\;&\le\;
e^{\varepsilon}\bigl(e^{\varepsilon}\Pr[M(D')\in S] + \delta\bigr) + \delta\\
&=\;
e^{2\varepsilon}\Pr[M(D')\in S] + (e^{\varepsilon}+1)\delta.
\end{align}
This is exactly $(2\varepsilon,(1+e^{\varepsilon})\delta)$-DP w.r.t.\ bounded adjacency.
\end{proof}
Note that the above conversion is completely black box. As a result, we can also plug in the privacy parameters obtained from Section~\ref{app:ext-more-datapoints} and obtain the corresponding DP guarantees under bounded adjacency.

\subsubsection{Other Threat Models}\label{app:other-threat-models}
\noindent\textbf{When the two servers are malicious during heavy hitter estimation}, the tagging server needs to provide a zero-knowledge proof on not altering the protocol, as is done in~\cite{nebula}. \textbf{When users who submit the vectors are malicious in the aggregation phase~\cite{secagg}}, or equivalently, malicious users (say, there are $n_\text{mal}$ of them) may submit indices computed from certain embedding vectors of interest, attempting to learn whether there exist other users with similar text embeddings (even such embeddings may not be frequent among the user population), we can make a similar argument. From the perspective of the targeted users, this act weakens their privacy protection by enforcing a smaller cut-off threshold of heavy-hitter identification---decreasing from $\tau$ to $\tau-n_\text{mal}$ (see~\eqref{eq:dp-fre-delta}). This can be resolved by using a MPC protocol to ensure that the random submission (determined by a random coin) is strictly executed. In this case, the malicious users' data serve as dummies to enhance the privacy protection. We leave further measures to defend against such malicious acts to future work. MPC protocols specifically designed for DP~\cite{fu2024benchmarking,olive2025,pine,finitedpmpc} or relaxed privacy frameworks~\cite{talwar2022differentialsecrecydistributeddata} may also be useful.

\section{Simulation-Based Security Sketch}
\label{app:semi-honest-security}

We provide a simulation-based security sketch for the distributed
implementation in Algorithm~\ref{alg:all}. This argument is not intended as a new
cryptographic primitive; rather, it shows that, under the same assumptions
as Section~\ref{sec:prob} and using the cited primitives as black boxes, the distributed
protocol reveals no information beyond the explicitly released differentially
private summary and the unavoidable metadata described below.

\vspace{0pt}\noindent\textbf{Adversarial model.}
Let \Ssyn denote the synthesis server and \Stag denote the tagging
server. We consider static probabilistic poly\-nomial-time semi-honest
adversaries: corrupted parties follow the protocol but may try to infer
additional information from their views. We assume \Ssyn and \Stag do
not collude. User-to-server submissions are sent through an anonymous
channel, modeled as an ideal shuffler that reveals only the multiset of
messages delivered to the recipient and not the sender identities or links
between a user's messages in different phases. We consider corruption of
\Ssyn, corruption of \Stag, or corruption of a subset of users. If one also
allows corrupted users to collude with \Ssyn, then the ideal leakage must
include the multiplicity of the VOPRF tags known to those corrupted users;
our main statement below excludes such user--server collusion.

\vspace{0pt}\noindent\textbf{Ideal functionality.}
Define an ideal functionality
\(\mathcal F_{\mathsf{D\text{-}FreE2T}}\) parameterized by the public
parameters
\[
(E,Z,L,\{o_j\}_{j=1}^k,p_s,\tau,\lambda,\gamma,\sigma,\beta).
\]
On inputs \(v_1,\ldots,v_N\), the functionality computes
\(x_i=E(v_i)\), the projected bucket index \(b_i\), the sampled participation
coin \(\mathsf{coin}_i\sim\mathrm{Bernoulli}(p_s)\), and an opaque tag
\(h_i\) for each bucket. Equal buckets receive equal tags, and distinct
buckets receive computationally independent tags. The functionality then
runs the same sample-and-threshold heavy-hitter step, including the dummy
submissions drawn from \(\mathsf{TSDLap}(\lambda,\gamma)\), and outputs the
heavy-hitter tag set \(H\) and counts \(C=\{c_h\}_{h\in H}\). For each
\(h\in H\), it computes the quantized, locally discrete-Gaussian-perturbed
sum
\[
T_h=\sum_{i:\,h_i=h,\,\mathsf{coin}_i=1}
       \bigl(Q_\beta(x_i)+\xi_i\bigr),
\]
where the coordinates of \(\xi_i\) are sampled independently from the
discrete Gaussian distribution used in Algorithm~2. It outputs the released
centroid
\[
\widehat \mu_h = \frac{\beta}{c_h} T_h .
\]

The public output is
\[
\mathsf{Out}=(H,C,\{\widehat\mu_h\}_{h\in H}).
\]
The synthesis server's ideal leakage additionally contains the
subthreshold count-of-counts
\[
\mathcal L_A^{<\tau}=\{m_q\}_{q=1}^{\tau-1},
\]
where \(m_q\) is the number of received non-heavy opaque tags with
multiplicity \(q\), after dummy insertion. This captures exactly the
information observed by \Ssyn from non-released pseudonymous tags; the
individual tag names themselves carry no semantic information. The
tagging server's ideal leakage consists of the number of VOPRF
invocations, the public output, and the public heavy-tag oracle induced by
its VOPRF key. This oracle captures the fact that a tagging server holding
the VOPRF key can test guessed bucket indices against the public released
heavy tags, but it still does not learn which users produced those tags.

\vspace{0pt}\noindent\textbf{Security statement.}
Assume that (i) the VOPRF realizes its standard oblivious and verifiable PRF
functionality, (ii) the anonymous channel realizes an ideal shuffler, and
(iii) the secure aggregation protocol realizes an ideal tagged-sum
functionality that reveals only the aggregate sum for each released tag.
Then, for every PPT semi-honest adversary \(\mathcal A\) corrupting either
\Ssyn, \Stag, or an arbitrary subset of users, there exists a PPT
simulator \(\mathsf{Sim}\) such that
\[
\mathsf{View}^{\Pi}_{\mathcal A}(v_1,\ldots,v_N)
\;\approx_c\;
\mathsf{Sim}\bigl(\mathsf{aux}_{\mathcal A},
                  \mathsf{Out},
                  \mathcal L_{\mathcal A}\bigr),
\]
where \(\Pi\) is Algorithm~2, \(\mathsf{aux}_{\mathcal A}\) contains the
inputs and randomness of corrupted users, and
\(\mathcal L_{\mathcal A}\) is the ideal leakage defined above for the
corrupted party type. The indistinguishability is computational, with
negligible distinguishing advantage in the cryptographic security parameter.

\vspace{0pt}\noindent\textbf{Proof sketch.}
We prove the statement by a sequence of standard hybrids.

\emph{Hybrid 0: the real execution.}
This is the execution of Algorithm~2. Each user locally computes
\(x_i=E(v_i)\), projects and buckets it to obtain \(b_i\), obtains a VOPRF
tag \(h_i\), submits \(h_i\) with probability \(p_s\), and, if \(h_i\in H\),
secret-shares a locally perturbed vector tagged by \(h_i\).

\emph{Hybrid 1: replace VOPRF by its ideal functionality.}
Replace the real VOPRF interaction between each user and \Stag with an
ideal VOPRF call. By the VOPRF security guarantee, \Stag's view of each
call can be simulated from the number of invocations: the blinded inputs are
indistinguishable from random protocol messages and reveal neither \(b_i\)
nor the output \(h_i\). Users with the same bucket obtain the same PRF tag,
whereas the tags for distinct buckets are computationally indistinguishable
from independent random strings to any party not holding the VOPRF key.
Verifiability guarantees that honest users accept only tags generated under
the committed VOPRF key. Therefore, this replacement changes the
adversary's view only negligibly.

\emph{Hybrid 2: replace user submissions by an ideal shuffler.}
Replace each user-to-server communication channel by an ideal anonymous
channel. The recipient learns only the multiset of submitted tags or tagged
shares, not the identities of the senders and not whether two messages in
different phases came from the same user. Thus, the synthesis server's
heavy-hitter-phase view is completely determined by the multiset
\[
M_{\mathsf{HH}}
=
\{h_i:\mathsf{coin}_i=1\}\cup \mathsf{Dummy},
\]
and the centroid-phase view is determined by the multiset of tagged shares
for tags in \(H\). This is exactly the metadata included in
\(\mathcal L_A^{<\tau}\), together with the released heavy tags \(H\) and
counts \(C\).

\emph{Hybrid 3: simulate non-heavy pseudonymous tags.}
For tags whose multiplicity is below \(\tau\), \Ssyn sees only opaque
pseudorandom strings and their multiplicities. Since the VOPRF tags for
unknown bucket indices are pseudorandom, the simulator samples fresh random
strings and assigns exactly \(m_q\) of them multiplicity \(q\), for each
\(q\in\{1,\ldots,\tau-1\}\). This produces the same distribution as the real
view up to the PRF distinguishing advantage. The dummy submissions are
accounted for in the same count-of-counts leakage. The DP analysis in
Section~\ref{sec:privacy-ddg} then bounds the information contained in
\(\mathcal L_A^{<\tau}\).

\emph{Hybrid 4: replace secure aggregation by ideal tagged sums.}
For each heavy tag \(h\in H\), let \(G_h\) be the set of sampled users with
tag \(h\). Each such user contributes an integer vector
\[
\widetilde x_i = Q_\beta(x_i)+\xi_i
\]
and sends additive secret shares of \(\widetilde x_i\) to the two
non-colluding servers. Consider a corrupted \Ssyn. The share of every
honest user's vector received by \Ssyn is uniformly random over the
aggregation group, independent of \(\widetilde x_i\). Given the ideal output
\(T_h=\sum_{i\in G_h}\widetilde x_i\), the simulator samples all shares
received by \Ssyn uniformly at random and programs the aggregate share sent
from \Stag to make the reconstructed sum equal \(T_h\). This is identically
distributed to the real execution because additive secret sharing is
perfectly private against either single server. The case of a corrupted
\Stag is symmetric. If an off-the-shelf SecAgg protocol is used instead of
two-server additive sharing, the same step follows from the simulation
security of the SecAgg functionality.

\emph{Hybrid 5: simulate users' views.}
A corrupted user sees only its own input, local randomness, VOPRF response,
participation coin, local noise, generated shares, the broadcast
\((H,C)\), and the public output. The simulator can reproduce this view
directly from the corrupted user's input and randomness plus
\(\mathsf{Out}\). Because users receive no other users' messages, no
additional simulation is required.

After these replacements, the transcript is generated entirely from the
corrupted parties' inputs and randomness, the public output
\(\mathsf{Out}\), and the leakage \(\mathcal L_{\mathcal A}\). Hence the
real and ideal executions are computationally indistinguishable.

The simulation argument shows that the cryptographic implementation does not
leak raw texts, bucket indices, or individual embeddings beyond the ideal
released objects above. Therefore the end-to-end privacy of the distributed
protocol reduces to the privacy of the ideal release. By Lemma~\ref{lem:hh-all} and
Theorem~\ref{thm:main-dgg}, the distributed heavy-hitter step and the distributed discrete
Gaussian centroid step satisfy the DP parameters analyzed in Section~\ref{sec:privacy-ddg}.
Composing these components and adding the high-probability sensitivity term yields Theorem~\ref{thm:dp-final}. 
\balance

\end{document}